\documentclass{article}
\usepackage{graphicx} 
\usepackage[a4paper,margin=2.5cm]{geometry}
\usepackage[T1]{fontenc}
\usepackage{lmodern}

\usepackage[T1]{fontenc}
\usepackage{lmodern}
\usepackage[a4paper,margin=2.5cm]{geometry}

\usepackage{cite}
\usepackage{amsmath,amssymb,amsfonts}
\usepackage{amsthm}
\usepackage{bm}

\usepackage{algorithm}
\usepackage{algorithmic}

\usepackage{graphicx}
\usepackage{xcolor}
\usepackage{tikz}
\usepackage{pgfplots}
\pgfplotsset{compat=1.18}

\usepackage{textcomp}
\usepackage{xurl}
\usepackage{hyperref}
\hypersetup{hidelinks}

\newtheorem{theorem}{Theorem}
\newtheorem{lemma}[theorem]{Lemma}
\newtheorem{corollary}[theorem]{Corollary}
\newtheorem{remark}[theorem]{Remark}
\newtheorem{proposition}[theorem]{Proposition}
\newtheorem{conjecture}[theorem]{Conjecture}

\def\BibTeX{{\rm B\kern-.05em{\sc i\kern-.025em b}\kern-.08em
    T\kern-.1667em\lower.7ex\hbox{E}\kern-.125emX}}

\title{Complexity, Bounds, and Exact Algorithms for Rainbow $k$-Domination in Regular Graphs}

\author{
Piotr Lange\\
\small Faculty of Electronics, Telecommunications and Informatics\\
\small Gdańsk University of Technology\\
\small Narutowicza 11/12, 80-233 Gdańsk, Poland\\
\small \texttt{piotr.lange@student.pg.edu.pl}
}

\date{September 2026}

\begin{document}

\maketitle

\begin{abstract}
A rainbow $k$-dominating function assigns the empty set or exactly one of the colors $\{1\},\{2\},\ldots,\{k\}$ to each vertex of a simple graph in such a way that every vertex receiving $\emptyset$ sees all $k$ colors in its neighborhood. This model has several facility-location interpretations where colored vertices represent sites hosting different types of services, and uncolored vertices represent the beneficiaries of those services who must have access to all of them. In this paper we study this domination variant in regular graphs. First, we develop lower and upper bounds for this class of graphs, and explore their structural consequences. Then, we use those results to prove NP-completeness of a decision formulation of this problem. Additionally, we derive exact formulas for the rainbow $k$-domination number of cycles, M\"obius ladders and prisms. Lastly, we develop a dynamic-programming algorithm for generalized Petersen graphs of the form $P(6t, t)$ and state two conjectures based on the output of its implementation and related integer linear programming computations.
\end{abstract}

\noindent\textbf{Keywords:}
rainbow domination, regular graphs, domination in graphs, structural bounds, probabilistic bounds, exact formulas, computational complexity, exact algorithms, generalized Petersen graphs, M\"obius ladders

\section{Introduction}
\label{sec:introduction}

Domination is one of the central concepts in graph theory and combinatorial optimization. A dominating set in a graph is a subset of its vertices such that every vertex outside of it is adjacent to at least one vertex inside it. This condition has numerous applications in facility location, system control and monitoring, network design, and resource distribution. Related research has led to the development of a large family of other domination-type graph parameters \cite{haynes1998fundamentals,haynes2020topics}.

Rainbow domination parameters extend this idea to systems where every vertex receiving no assignment must have all required types of services, resources, or labels in its neighborhood. Two closely related but different rainbow domination models have been introduced in the literature. The classical $k$-rainbow domination problem, introduced by Bre\v{s}ar, Henning, and Rall, assumes that an arbitrary subset of the color set $[k]=\{1,\ldots,k\}$ is assigned to each vertex so that every vertex receiving the empty set must see all $k$ colors in its open neighborhood \cite{bresar2008rainbow}. Thus, this model allows for the assignment of several colors to one vertex.

The parameter studied in this paper is the singleton version. A rainbow $k$-dominating function assigns exactly one color to each selected vertex of a graph such that every uncolored vertex has all $k$ colors in its open neighborhood. This model has appeared in the literature as singleton $k$-rainbow domination and, for $k=2$, under the names rainbow $2$-domination and dual-server domination \cite{bonomo2018domination, erves2021petersen, chellali2026dual}. Throughout this paper, we use the term rainbow $k$-domination and denote the minimum weight of a rainbow $k$-dominating function of a graph $G$ by $\widetilde{\gamma}_{k}(G)$.

The restriction to only one color being assigned per selected vertex is natural in settings where each selected location can accommodate or provide only one type of resource. One such application, concerning the spatial modeling of post-disaster shelter camps, was developed by G\"oz\"upek, Lange, and Raczek in \cite{gozupek2026shelter}. In their model, uncolored vertices represent accommodation units, while colored vertices represent service-providing facilities. Each facility type is associated with a different color, and an uncolored vertex is adequately served only when all required facility types are present in its neighborhood. They further compared the use of rainbow and Roman domination parameters for facility placement constraints.

The case $k=2$ was studied by Bonomo, Bre\v{s}ar, Grippo, Milani\v{c}, and Safe in their work on domination parameters involving the number $2$ \cite{bonomo2018domination}. In particular, they showed that rainbow $2$-domination remains computationally hard even when restricted to split graphs. The parameter was also studied for larger values of $k$ and specific graph families in several papers by \v{Z}erovnik and coauthors \cite{erves2021petersen, brezovnik2025rainbow, zerovnik2026singleton, zerovnik2026four}.

In this paper, we focus on regular graphs, first developing general bounds for this class and recalling a characterization of the extremal equality case. We then present the algorithmic consequences of this characterization, proving NP-completeness of the corresponding decision problem. Next, we apply the bounds to obtain exact formulas for several standard families, including cycles, M\"obius ladders, and prism graphs. Finally, we present a specialized dynamic programming algorithm for generalized Petersen graphs and report computational evidence leading to further conjectures.

The paper is organized as follows. Section~\ref{sec:preliminaries} introduces notation and definitions. General bounds for regular graphs and the characterization of the extremal equality case are presented in Section~\ref{sec:bounds}. Section~\ref{sec:complexity} contains the NP-completeness proof for regular graphs. Exact results for cycles, M\"obius ladders, and prism graphs are given in Section~\ref{sec:exact-families}. Section~\ref{sec:dp} presents the dynamic programming algorithm for generalized Petersen graphs. Computational evidence and conjectures are discussed in Section~\ref{sec:computations}. Section~\ref{sec:conclusion} concludes the paper.

\section{Preliminaries}
\label{sec:preliminaries}

Throughout the paper, all graphs are finite, simple, and undirected. For a graph $G$, we denote its vertex set and edge set by $V(G)$ and $E(G)$, respectively. For two disjoint vertex sets $X,Y\subseteq V(G)$, we write $e(X,Y)$ for the number of edges with one endpoint in $X$ and the other endpoint in $Y$. The order of $G$ is $n=|V(G)|$. For a vertex $v\in V(G)$, its open neighborhood is denoted by $N_G(v)$, or by $N(v)$ when the graph is clear from the context. The degree of a vertex $v$ is $d_G(v)=|N_G(v)|$. The minimum and maximum degrees of $G$ are denoted by $\delta(G)$ and $\Delta(G)$, respectively. A graph is $r$-regular if every vertex has degree $r$. A $3$-regular graph is alternatively called a cubic graph.

For a positive integer $k$, let $[k]=\{1,2,\ldots,k\}$. A rainbow $k$-dominating function of a graph $G$ is a function
\[
f:V(G)\to \{\emptyset,\{1\},\{2\},\ldots,\{k\}\}
\]
such that, for every vertex $v\in V(G)$ with $f(v)=\emptyset$,
\[
\bigcup_{u\in N(v)} f(u)=[k].
\]
The weight of $f$ is
\[
w(f)=|\{v\in V(G):f(v)\neq \emptyset\}|.
\]
Let $\mathcal{R}_k(G)$ be the set of all rainbow $k$-dominating functions of $G$. Then, the rainbow $k$-domination number of $G$ can be defined as follows.
\[
\widetilde{\gamma}_{k}(G) = \min_{f\in \mathcal{R}_k(G)} w(f).
\]

A graph and two of its rainbow $3$-dominating functions, one of them minimum, are shown in Figure~\ref{fig:rainbow-example}.

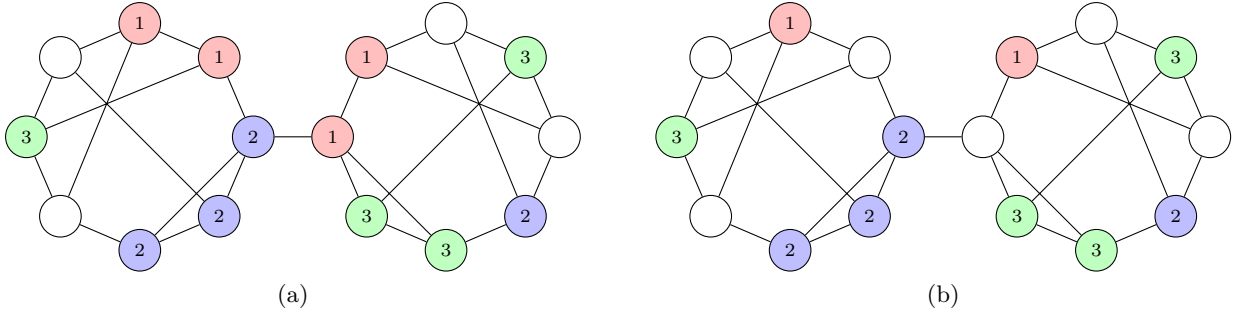
\begin{figure}[t]
\centering

\begin{tikzpicture}[
    x=0.75cm,
    y=0.75cm,
    vertex/.style={
        circle,
        draw,
        minimum size=5.5mm,
        inner sep=0pt,
        font=\scriptsize
    },
    c1/.style={vertex,fill=red!25},
    c2/.style={vertex,fill=blue!25},
    c3/.style={vertex,fill=green!25},
    empty/.style={vertex,fill=white}
]


\begin{scope}[xshift=-4.3cm]

\node[c1] (a0) at ( 0.7, 0.0) {1};
\node[c2] (a1) at (-1.3,-1.4) {2};
\node[c3] (a2) at ( 4.1, 1.4) {3};
\node[c1] (a3) at (-2.7, 2.0) {1};
\node[c3] (a4) at (-4.7, 0.0) {3};
\node[c2] (a5) at ( 4.1,-1.4) {2};
\node[c1] (a6) at ( 1.3, 1.4) {1};
\node[c2] (a7) at (-2.7,-2.0) {2};
\node[c3] (a8) at ( 1.3,-1.4) {3};
\node[c1] (a9) at (-1.3, 1.4) {1};
\node[c3] (a10) at ( 2.7,-2.0) {3};
\node[empty] (a11) at ( 2.7, 2.0) {};
\node[empty] (a12) at ( 4.7, 0.0) {};
\node[empty] (a13) at (-4.1, 1.4) {};
\node[empty] (a14) at (-4.1,-1.4) {};
\node[c2] (a15) at (-0.7, 0.0) {2};

\draw
(a0)--(a6) (a0)--(a8) (a0)--(a10) (a0)--(a15)
(a1)--(a7) (a1)--(a13) (a1)--(a15)
(a2)--(a8) (a2)--(a11) (a2)--(a12)
(a3)--(a9) (a3)--(a13) (a3)--(a14)
(a4)--(a9) (a4)--(a13) (a4)--(a14)
(a5)--(a10) (a5)--(a11) (a5)--(a12)
(a6)--(a11) (a6)--(a12)
(a7)--(a14) (a7)--(a15)
(a8)--(a10)
(a9)--(a15);

\node[font=\small] at (0,-2.8) {(a)};

\end{scope}

\begin{scope}[xshift=4.3cm]

\node[empty] (b0) at ( 0.7, 0.0) {};
\node[c2] (b1) at (-1.3,-1.4) {2};
\node[c3] (b2) at ( 4.1, 1.4) {3};
\node[c1] (b3) at (-2.7, 2.0) {1};
\node[c3] (b4) at (-4.7, 0.0) {3};
\node[c2] (b5) at ( 4.1,-1.4) {2};
\node[c1] (b6) at ( 1.3, 1.4) {1};
\node[c2] (b7) at (-2.7,-2.0) {2};
\node[c3] (b8) at ( 1.3,-1.4) {3};
\node[empty] (b9) at (-1.3, 1.4) {};
\node[c3] (b10) at ( 2.7,-2.0) {3};
\node[empty] (b11) at ( 2.7, 2.0) {};
\node[empty] (b12) at ( 4.7, 0.0) {};
\node[empty] (b13) at (-4.1, 1.4) {};
\node[empty] (b14) at (-4.1,-1.4) {};
\node[c2] (b15) at (-0.7, 0.0) {2};

\draw
(b0)--(b6) (b0)--(b8) (b0)--(b10) (b0)--(b15)
(b1)--(b7) (b1)--(b13) (b1)--(b15)
(b2)--(b8) (b2)--(b11) (b2)--(b12)
(b3)--(b9) (b3)--(b13) (b3)--(b14)
(b4)--(b9) (b4)--(b13) (b4)--(b14)
(b5)--(b10) (b5)--(b11) (b5)--(b12)
(b6)--(b11) (b6)--(b12)
(b7)--(b14) (b7)--(b15)
(b8)--(b10)
(b9)--(b15);

\node[font=\small] at (0,-2.8) {(b)};

\end{scope}

\end{tikzpicture}

\caption{A graph $G$ and two of its rainbow $3$-dominating functions. Numbers inside colored vertices indicate the assigned colors, and white vertices are uncolored. (a) A rainbow $3$-dominating function of weight $12$. (b) A minimum rainbow $3$-dominating function of weight $10$.}
\label{fig:rainbow-example}
\end{figure}

Even though formally a rainbow $k$-dominating function $f$ assigns empty or one-element sets to the vertices of $G$, for simplicity we equivalently say that $f$ colors some vertices and leaves the rest uncolored. We write
\[
C_f=\{v\in V(G):f(v)\neq\emptyset\}
\]
for the set of colored vertices and
\[
U_f=\{v\in V(G):f(v)=\emptyset\}
\]
for the set of uncolored vertices. Alternatively, we write $C$ and $U$ when the function is clear from the context. For each color $i\in[k]$, we also write
\[
C_i=\{v\in V(G):f(v)=\{i\}\}.
\]
Thus
\[
C=C_1\dot\cup C_2\dot\cup\cdots\dot\cup C_k.
\]
Then $w(f)=|C|$.

For a vertex $v$ of $G$, a function $f: V(G)\to \{\emptyset,\{1\},\ldots,\{k\}\}$, and a color $i\in[k]$, we say that $v$ sees the color $i$ in its neighborhood if $i \in \bigcup_{u\in N(v)} f(u)$. In an analogous setting, if $h: V(G)\to [k]$ is a coloring with $k$ colors and $h$ assigns the color $i\in[k]$ to at least one neighbor of $v\in V(G)$, we say that $v$ sees this color in its neighborhood.

Recall that in the classical $k$-rainbow domination model, a function
\[
g:V(G)\to 2^{[k]}
\]
may assign an arbitrary subset of $[k]$ to each vertex. Its weight is
\[
w(g)=\sum_{v\in V(G)}|g(v)|.
\]
If $\bigcup_{u\in N(v)} g(u)=[k]$ for every vertex $v\in V(G)$ with $g(v)=\emptyset$, then $g$ is a $k$-rainbow dominating function of $G$. The minimum weight of such a function is called the $k$-rainbow domination number of $G$ and denoted by $\gamma_{rk}(G)$.

We record the following relation between the two rainbow parameters, observed in \cite{erves2021petersen}. It follows from the fact that every rainbow $k$-dominating function of a graph is also a $k$-rainbow dominating function of the same graph and of the same weight.
\begin{proposition}
\label{prop:rainbow-relation}
Let $G$ be a graph. Then, for every positive integer $k$,
\[
\widetilde{\gamma}_{k}(G) \geq \gamma_{rk}(G).
\]
\end{proposition}

For a fixed positive integer $k$, the corresponding decision problem is the following.

\smallskip
\noindent
\textsc{Rainbow $k$-Domination}

\smallskip
\noindent
\textbf{Instance:} A graph $G$ and an integer $B$.

\noindent
\textbf{Question:} Is $\widetilde{\gamma}_{k}(G)\le B$?
\smallskip

The problem is in NP. Indeed, for any function $f:V(G)\to\{\emptyset,\{1\},\ldots,\{k\}\}$, its weight and the rainbow $k$-domination condition can be verified in polynomial time.

If $G$ is $r$-regular and $k>r$, then no uncolored vertex can see all $k$ colors in its open neighborhood. Hence, any rainbow $k$-dominating function of $G$ must color all vertices. Therefore, in this case
\[
\widetilde{\gamma}_{k}(G)=|V(G)|.
\]
Thus, when studying $r$-regular graphs, we usually restrict our attention to $1\le k\le r$.

We will also use the following sparse induced-subgraph parameter. For an integer $s\ge 0$, let $\alpha_s(G)$ denote the maximum order of an induced subgraph of $G$ with maximum degree at most $s$. Equivalently,
\[
\alpha_s(G)=\max\{|X|:X\subseteq V(G)\text{ and }\Delta(G[X])\le s\}.
\]
In particular, $\alpha_0(G)=\alpha(G)$ is the independence number of $G$.

\section{Bounds and extremal regular graphs}
\label{sec:bounds}

In this section we present several bounds for rainbow $k$-domination in regular graphs. We begin with a lower bound derived from a general degree-based bound proved in \cite{gozupek2026shelter}. Then, we record a second lower bound expressed in terms of sparse induced subgraphs. To complement them, we also establish two probabilistic upper bounds. Finally, we reformulate the characterization of the equality case when $k=r$ and $\widetilde{\gamma}_{k}(G)=n/2$ in the rainbow $k$-domination model.

\begin{corollary}
\label{cor:regular-lower}
Let $G$ be an $r$-regular graph of order $n$, and let $1\le k\le r$. Then
\[
\widetilde{\gamma}_{k}(G) \ge kq+\min\{k,s\},
\]
where $q$ and $s$ are integers satisfying
\[
n=q(k+r)+s, \qquad 0\le s<k+r.
\]
\end{corollary}
\begin{proof}
Since $G$ is $r$-regular, $\Delta(G)=r$ and the result follows directly from the general bound proved in \cite{gozupek2026shelter}.
\end{proof}

A cubic graph has an even number of vertices. Hence, Corollary~\ref{cor:regular-lower} with $r=k=3$ gives the following bound.

\begin{corollary}
\label{cor:cubic-extremal}
Let $G$ be a cubic graph of order $n$. Then
\[
\widetilde{\gamma}_{3}(G)\ge
\begin{cases}
\frac{n}{2}, & \text{if } n\equiv 0\pmod 6,\\
\frac{n}{2}+1, & \text{if } n\equiv 2,4\pmod 6.
\end{cases}
\]
\end{corollary}

The following proposition uses a different structural argument and can give a stronger bound for graphs with small values of $\alpha_{r-k}(G)$ compared with Corollary~\ref{cor:regular-lower}. It can be considered a consequence of an analogous statement proven for $k$-domination in \cite{hansberg2013kdomination}.

\begin{proposition}
\label{prop:sparse-lower}
Let $G$ be an $r$-regular graph of order $n$, and let $1\le k\le r$. Then
\[
\widetilde{\gamma}_{k}(G)\ge n-\alpha_{r-k}(G).
\]
In particular, if $k=r$ then
\[
\widetilde{\gamma}_{r}(G)\ge n-\alpha(G).
\]
\end{proposition}
\begin{proof}
Let $f$ be a minimum rainbow $k$-dominating function of $G$, and let
$C$ and $U$ be its sets of colored and uncolored vertices, respectively. Since $G$ is $r$-regular and every uncolored vertex $u\in U$ of $G$ must have at least $k$ colored neighbors, it follows that
\[
|N(u)\cap U|\le r-k
\]
and therefore
\[
\Delta(G[U])\le r-k.
\]
By the definition of $\alpha_{r-k}(G)$,
\[
|U|\le \alpha_{r-k}(G)
\]
and finally
\[
\widetilde{\gamma}_{k}(G)=|C|=n-|U|\ge n-\alpha_{r-k}(G).
\]
\end{proof}

We now derive two complementary upper bounds using a probabilistic approach inspired by similar work regarding the classical domination number, in particular \cite{jafarirad2019probabilistic} and \cite{henning2026probabilistic}. First, we present a bound obtained by substituting $\delta(G)=r$ to a general upper bound developed in \cite{gozupek2026shelter}.

\begin{corollary}
\label{cor:regular-upper}
Let $G$ be an $r$-regular graph of order $n$, and let $1\le k\le r$. Then
\[
\widetilde{\gamma}_k(G) \leq n\min\left\{1, \frac{k\bigl(\ln(r+k)+1\bigr)}{r+k}\right\}.
\]
\end{corollary}

We also give the following bound derived specifically for regular graphs, which is stronger for certain values of $r$ and $k$.

\begin{theorem}
\label{thm:regular-upper}
Let $G$ be an $r$-regular graph of order $n$, and let $1\le k\le r$. Then
\[
\widetilde{\gamma}_{k}(G) \leq n\min\left\{1,\,\frac{r+k\left(1-\frac{1}{k}\right)^r}{r+1}\right\}.
\]
\end{theorem}

\begin{proof}
Recall that for any graph $H$,
\[
\alpha(H)\ge \frac{|V(H)|}{\Delta(H)+1}.
\]
Since $\Delta(G)=r$, there exists an independent set $U$ such that $|U|\ge n / (r+1)$. Let $C=V(G)\setminus U$. Then every vertex of $U$ has exactly $r$ neighbors in $C$.

Assign independently and uniformly one of the $k$ colors to every vertex of $C$. Fix a vertex $u\in U$ and a color $i\in[k]$. The probability that the vertex $u$ does not see the color $i$ in its neighborhood is $(1-1/k)^r$. Hence, by the union bound, the probability that at least one color is missing from $N(u)$ is at most $k(1-1/k)^r$. If $k\left(1- (1 / k)\right)^r\ge 1$, then
\[
\frac{r+k\left(1-\frac{1}{k}\right)^r}{r+1}\ge 1,
\]
and the result follows from the trivial bound $\widetilde{\gamma}_k(G)\le n$.
Suppose otherwise. Then, there exists a function $f_0$ that assigns the $k$ colors to the vertices of $C$ such that at least 
\[
|U|\left(1-k\left(1-\frac{1}{k}\right)^r\right)
\]
vertices of $U$ see all $k$ colors in their neighborhoods.

Construct a function $f:V(G)\to\{\emptyset,\{1\},\ldots,\{k\}\}$ such that
\[
f(v)=\begin{cases}
f_0 (v), & \text{if } v \in C,\\
\{1\}, & \text{if } v\in U \text{ and } \bigcup_{u\in N(v)} f_0(u) \neq [k],\\
\emptyset, & \text{if } v\in U \text{ and } \bigcup_{u\in N(v)} f_0(u) = [k].
\end{cases}
\]
Since all vertices being assigned an empty set by $f$ see all $k$ colors in their neighborhoods, $f$ is a rainbow $k$-dominating function. Hence,
\[
\widetilde{\gamma}_{k}(G)\le n-\frac{n}{r+1}\left(1-k\left(1-\frac{1}{k}\right)^r\right)
=n\frac{r+k\left(1-\frac{1}{k}\right)^r}{r+1}.
\]
Combining this with the trivial bound $\widetilde{\gamma}_{k}(G)\le n$ completes the proof.
\end{proof}

Let
\[
\rho(r,k)=\min\left\{1,\frac{k\bigl(\ln(r+k)+1\bigr)}{r+k},\frac{r+k\left(1-\frac{1}{k}\right)^r}{r+1}\right\}
\]
be the normalized value of the combination of the upper bounds from Corollary~\ref{cor:regular-upper} and Theorem~\ref{thm:regular-upper}. Its numerical value for different values of $r$ and $k$ is shown in Figure~\ref{fig:probabilistic-bound-plot}. The combined upper bound benefits from the bound from Theorem~\ref{thm:regular-upper} at smaller values of $r$ and from the asymptotically stronger logarithmic bound from Corollary~\ref{cor:regular-upper} at larger values of $r$.

\begin{figure}[ht]
\centering
\begin{tikzpicture}
\begin{axis}[
    width=0.65\columnwidth,
    height=0.45\columnwidth,
    xlabel={$r$},
    ylabel={$\rho(r,k)$},
    xmin=3, xmax=40,
    ymin=0.2, ymax=1.02,
    xtick={3,5,10,15,20,25,30,35,40},
    ytick={0.2,0.4,0.6,0.8,1.0},
    grid=both,
    minor tick num=1,
    legend style={
        at={(0.97,0.97)},
        anchor=north east,
        font=\footnotesize
    },
    every axis plot/.append style={
        thick,
        mark=*,
        mark size=1.4pt
    }
]

\addplot[
    blue,
    samples=38,
    domain=3:40
]
{min(1,min(2*(ln(x+2)+1)/(x+2),(x+2*pow(1-1/2,x))/(x+1)))};
\addlegendentry{$k=2$}

\addplot[
    red,
    samples=38,
    domain=3:40
]
{min(1,min(3*(ln(x+3)+1)/(x+3),(x+3*pow(1-1/3,x))/(x+1)))};
\addlegendentry{$k=3$}

\addplot[
    green!50!black,
    samples=37,
    domain=4:40
]
{min(1,min(4*(ln(x+4)+1)/(x+4),(x+4*pow(1-1/4,x))/(x+1)))};
\addlegendentry{$k=4$}

\addplot[
    orange,
    samples=36,
    domain=5:40
]
{min(1,min(5*(ln(x+5)+1)/(x+5),(x+5*pow(1-1/5,x))/(x+1)))};
\addlegendentry{$k=5$}

\addplot[
    violet,
    samples=35,
    domain=6:40
]
{min(1,min(6*(ln(x+6)+1)/(x+6),(x+6*pow(1-1/6,x))/(x+1)))};
\addlegendentry{$k=6$}

\end{axis}
\end{tikzpicture}
\caption{The normalized combined probabilistic upper bound $\rho(r,k)$ for different values of $r$ and $k$.}
\label{fig:probabilistic-bound-plot}
\end{figure}
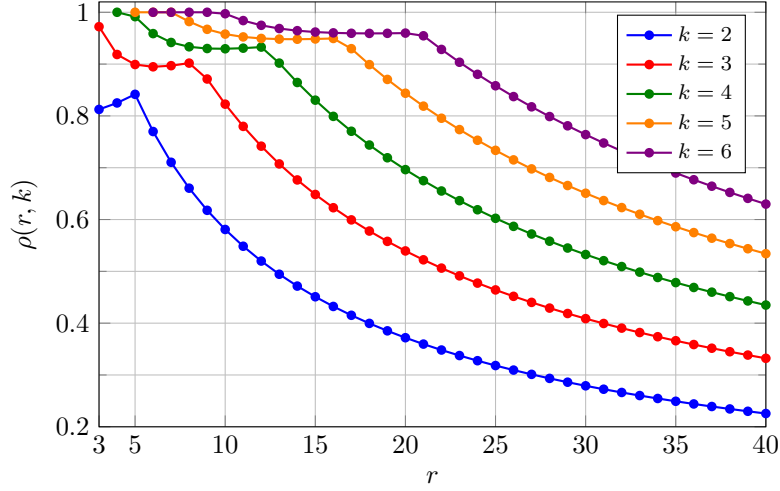

The following characterization is a singleton reformulation of the equality case established by Kuzman for $k$-rainbow domination \cite[Theorem~3.1 and its proof]{kuzman2020}. We present a self-contained proof utilizing rainbow $k$-domination terminology and bounds.

\begin{theorem}
\label{thm:extremal-regular}
Let $r\ge 2$ and let $G$ be an $r$-regular graph on $n$ vertices. Then
\[
    \widetilde{\gamma}_{r}(G)=\frac{n}{2}
\]
if and only if there exist a bipartition $(U,C)$ of $G$ and a partition
\[
    C=C_1\dot\cup\cdots\dot\cup C_r
\]
such that for every $u\in U$ we have
\[
    |N(u)\cap C_i|=1
    \qquad \text{for all } i\in\{1,\dots,r\}.
\]
\end{theorem}

\begin{proof}
\noindent(\,$\Leftarrow$\,)
Assume that $G$ is bipartite with bipartition $(U,C)$ and that
\[
C=C_1\dot\cup\cdots\dot\cup C_r
\]
satisfies
\[
|N(u)\cap C_i|=1
\qquad \text{for every } u\in U \text{ and every } i\in\{1,\dots,r\}.
\]
Define $f$ by assigning color $i$ to every vertex of $C_i$ and leaving every vertex of $U$ uncolored. Then each $u\in U$ has exactly one neighbor in each $C_i$, hence it sees all $r$ colors in its neighborhood. Thus $f$ is a rainbow $r$-dominating function of $G$. Its weight is
\[
w(f)=|C|.
\]

Since $G$ is $r$-regular and bipartite with bipartition $(U,C)$, we have
\[
r|U|=|E(G)|=r|C|,
\]
and therefore $|U|=|C|=n/2$. Hence
\[
\widetilde{\gamma}_{r}(G)\le w(f)=|C|=\frac{n}{2}.
\]
By Corollary~\ref{cor:regular-lower},
\[
\widetilde{\gamma}_{r}(G)\ge \left\lceil \frac{n}{2}\right\rceil \ge \frac{n}{2}.
\]
Therefore,
\[
\widetilde{\gamma}_{r}(G)=\frac{n}{2}.
\]

\smallskip
\noindent(\,$\Rightarrow$\,)
Assume now that
\[
\widetilde{\gamma}_{r}(G)=\frac{n}{2},
\]
and let $f$ be a minimum rainbow $r$-dominating function on $G$. Define
\[
C=\{v\in V(G): f(v)\neq \emptyset\},
\qquad
U=\{v\in V(G): f(v)=\emptyset\}.
\]
Then
\[
|C|=\frac{n}{2}
\qquad \text{and} \qquad
|U|=\frac{n}{2}.
\]

Since $G$ is $r$-regular and $k=r$, every vertex $u\in U$ has exactly $r$ neighbors, and in order to see all $r$ colors under a singleton assignment, all these neighbors must be colored and must receive pairwise distinct colors. In particular, no vertex of $U$ is adjacent to another vertex of $U$, so $U$ is independent. Therefore,
\[
e(U,C)=r|U|.
\]

On the other hand, every vertex of $C$ has degree $r$, so
\[
e(U,C)\le r|C|.
\]
Since $|U|=|C|$, both sides are equal, and consequently every vertex of $C$ has all its neighbors in $U$. Thus $C$ is independent as well, and $G$ is bipartite with bipartition $(U,C)$.

For each $i\in\{1,\dots,r\}$ define
\[
C_i=\{v\in C: f(v)=\{i\}\}.
\]
Then
\[
C=C_1\dot\cup\cdots\dot\cup C_r.
\]
As observed above, every vertex $u\in U$ has exactly $r$ colored neighbors with pairwise distinct colors, so for every $i\in\{1,\dots,r\}$ we have
\[
|N(u)\cap C_i|=1.
\]

Finally, counting edges between $U$ and $C_i$ in two ways gives
\[
e(U,C_i)=|U|
\]
because each vertex of $U$ has exactly one neighbor in $C_i$, and also
\[
e(U,C_i)=r|C_i|
\]
because every vertex of $C_i$ has degree $r$ and all its neighbors lie in $U$. Hence
\[
r|C_i|=|U|=\frac{n}{2},
\]
so
\[
|C_i|=\frac{n}{2r}
\qquad \text{for every } i\in\{1,\dots,r\}.
\]
In particular, $2r$ divides $n$.
\end{proof}

\begin{remark}
\label{rem:cubic-extremal}
For cubic graphs, Theorem~\ref{thm:extremal-regular} shows that the extremal case
\[
\widetilde{\gamma}_{3}(G)=\frac{n}{2}
\]
is highly restrictive. Indeed, equality forces $6\mid n$, $G$ to be bipartite, and one side of the bipartition to admit a partition
\[
    C=C_1\dot\cup C_2\dot\cup C_3
\]
with
\[
|C_1|=|C_2|=|C_3|=\frac{n}{6},
\]
such that every vertex of the opposite side has exactly one neighbor in each $C_i$. Thus the colored vertices must form three equally large classes such that every vertex of the uncolored side has exactly one neighbor in each class.
\end{remark}

\section{NP-completeness for regular graphs}
\label{sec:complexity}

In this section we prove that, for every fixed $r\ge3$, deciding whether $\widetilde{\gamma}_r(G)=|V(G)|/2$ remains NP-complete even for bipartite $r$-regular graphs. For the reduction, we use the following problem.

\smallskip
\noindent
\textsc{Strong $r$-Coloring of Uniform Regular Hypergraphs}

\smallskip
\noindent
\textbf{Instance:} A simple linear $r$-uniform $r$-regular hypergraph $\mathcal H$.

\noindent
\textbf{Question:} Does $\mathcal H$ admit a strong coloring with $r$ colors?
\smallskip

For a hypergraph $\mathcal H=(X,\mathcal E)$, a strong $r$-coloring is a function $\chi:X\to[r]$ such that
\[
    \{\chi(x):x\in F\}=[r] \qquad\text{for every }F\in\mathcal E.
\]

Lange proved \textsc{Strong $r$-Coloring of Uniform Regular Hypergraphs} to be NP-complete for every fixed $r\ge3$ on simple linear $r$-uniform $r$-regular hypergraphs \cite{lange2026strong}.

First, we introduce a parametrized forcing gadget $P_r$ used in the reduction. Define two disjoint sets of vertices
\[
A=\{a_1,\ldots,a_r\},\qquad B=\{b_1,\ldots,b_r\}.
\]
Let $W=A\cup B$ and
\[
\begin{aligned}
    E_0&=A,\\
    E_i&=(A\setminus\{a_i\})\cup\{b_i\},
        &&1\le i\le r-2,\\
    E_{r-1}&=(A\setminus\{a_{r-1}\})\cup\{b_1\}.
\end{aligned}
\]
For $0\le i\le r-1$, put $\overline E_i=W\setminus E_i$. Define a set of hyperedges
\[
    \mathcal F=\{E_i,\overline E_i:0\le i\le r-1\},
\]
and let $\mathcal Z = (W, \mathcal F)$ be a hypergraph. Finally, let $P_r$ be the incidence graph of $\mathcal Z$. Figure~\ref{fig:forcing-gadget} illustrates the gadget for $r=3$.

\begin{figure}[t]
\centering
\begin{tikzpicture}[
    every node/.style={font=\small},
    vtx/.style={circle, draw, minimum size=6mm, inner sep=0pt},
    port/.style={circle, draw=red, minimum size=6mm, inner sep=0pt},
    lab/.style={font=\footnotesize}
]

\node[vtx] (a1) at (0,5) {$a_1$};
\node[vtx] (a2) at (0,4) {$a_2$};
\node[vtx] (a3) at (0,3) {$a_3$};
\node[vtx] (b1) at (0,2) {$b_1$};
\node[vtx] (b2) at (0,1) {$b_2$};
\node[port] (b3) at (0,0) {$b_3$};

\node[vtx] (E0) at (5,5) {$E_0$};
\node[vtx] (E1) at (5,4) {$E_1$};
\node[vtx] (E2) at (5,3) {$E_2$};
\node[port] (_E0) at (5,2) {$\overline E_0$};
\node[vtx] (_E1) at (5,1) {$\overline E_1$};
\node[vtx] (_E2) at (5,0) {$\overline E_2$};

\node[lab] at (0,5.8) {$W=A\cup B$};
\node[lab] at (5,5.8) {$\mathcal F$};

\draw (a1)--(E0);
\draw (a1)--(E2);
\draw (a1)--(_E1);

\draw (a2)--(E0);
\draw (a2)--(E1);
\draw (a2)--(_E2);

\draw (a3)--(E0);
\draw (a3)--(E1);
\draw (a3)--(E2);

\draw (b1)--(E1);
\draw (b1)--(E2);
\draw (b1)--(_E0);

\draw (b2)--(_E0);
\draw (b2)--(_E1);
\draw (b2)--(_E2);

\draw[red] (b3)--(_E0);
\draw (b3)--(_E1);
\draw (b3)--(_E2);

\end{tikzpicture}
\caption{The forcing gadget $P_r$ for $r=3$. The distinguished edge $b_3\overline E_0$ is highlighted in red.}
\label{fig:forcing-gadget}
\end{figure}
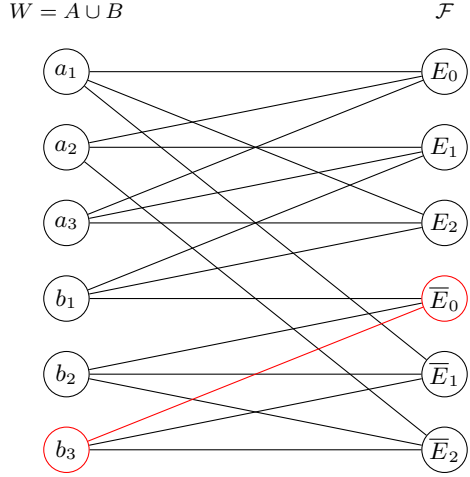

The $2r$ hyperedges of $\mathcal Z$ are distinct and contain $r$ vertices each. Since every vertex of $W$ belongs to exactly one member of each pair $\{E_i,\overline E_i\}$, $P_r$ is $r$-regular and bipartite. The two partite sets, each of size $2r$, are exactly $W$ and $\mathcal F$. Moreover, $P_r$ is connected. Indeed, $E_0$ joins all vertices of $A$, $\overline E_0$ joins all vertices of $B$, and $E_1$ contains vertices of both $A$ and $B$.

\begin{lemma}
\label{lem:orientation-gadget}
The gadget $P_r$ has the following properties.
\begin{enumerate}
    \item There exists a coloring $\psi: \mathcal F \to [r]$ such that every vertex of $W$ sees all $r$ colors.
    \item No coloring $\psi':W \to [r]$ exists such that every vertex of $\mathcal F\setminus\{\overline E_0\}$ sees all $r$ colors.
\end{enumerate}
\end{lemma}
\begin{proof}
For the first statement, let $\psi(E_i)=\psi(\overline E_i)=i+1$. Then every vertex of $W$ has exactly one neighbor of each color.

For the second statement, suppose that $W$ has a coloring $\psi':W \to [r]$ such that every vertex of $\mathcal F\setminus\{\overline E_0\}$ sees all $r$ colors. Then, all vertices of $E_0 = \{a_1,\ldots,a_r\}$ must be assigned pairwise distinct colors. Moreover, the requirement that all $r$ colors appear on $E_1$ forces $\psi'(b_1)=\psi'(a_1)$, while the same requirement for $E_{r-1}$ forces $\psi'(b_1)=\psi'(a_{r-1})$. This is impossible for $r\ge3$.
\end{proof}

Introduce a new vertex $x$ and obtain $P'_r$ from $P_r$ by deleting $b_r\overline E_0$ and adding $b_rx$. Then the following statement holds.

\begin{lemma}
\label{lem:color-transmission}
Let $c: \mathcal F\cup\{x\} \to [r]$ be a coloring such that every vertex of $W$ sees all $r$ colors in $P'_r$. Then $c(x)=c(\overline E_0)$.
\end{lemma}
\begin{proof}
Recall that
\[
\begin{aligned}
    N_{P'_r}(b_{r-1})&=\{\overline E_0,\ldots,\overline E_{r-1}\},\\
    N_{P'_r}(b_r)&=\{x,\overline E_1,\ldots,\overline E_{r-1}\}.
\end{aligned}
\]
These two neighborhoods have $r-1$ common vertices, and $c$ assigns them pairwise distinct colors. Hence, the two vertices belonging to only one of the discussed neighborhoods, namely $x$ and $\overline E_0$, must be assigned the same color under $c$.
\end{proof}

\begin{figure}[ht]
\centering
\resizebox{0.7\columnwidth}{!}{%
\begin{tikzpicture}[
    every node/.style={font=\small},
    vtx/.style={
        circle,
        draw,
        minimum size=6mm,
        inner sep=0pt
    },
    box/.style={
        draw,
        rounded corners,
        inner sep=4pt
    },
    del/.style={
        red,
        dashed,
        thick
    },
    add/.style={
        blue,
        thick
    }
]

\node[box, minimum width=4.2cm, minimum height=2.8cm] at (0,0) {};
\node at (0,1.1) {$I_j$};

\node[vtx] (xj) at (-1.0,0) {$x_j$};
\node[vtx] (Fj) at (1.0,0) {$F_j$};

\node[font=\footnotesize] at (-1.0,-0.8) {$x_j\in X_j$};
\node[font=\footnotesize] at (1.0,-0.8) {$F_j\in\mathcal E_j$};

\draw[del] (xj)--(Fj);

\node[box, minimum width=4.2cm, minimum height=2.8cm] at (6,0) {};
\node at (6,1.1) {$P_r^j$};

\node[vtx] (E0) at (5.0,0) {$\overline E_0^j$};
\node[vtx] (br) at (7.0,0) {$b_r^j$};

\node[font=\footnotesize] at (5.0,-0.8)
{$\overline E_0^j\in\mathcal F^j$};

\node[font=\footnotesize] at (7.0,-0.8)
{$b_r^j\in W^j$};

\draw[del] (E0)--(br);

\draw[add] (xj) to[bend left=15] (br);
\draw[add] (Fj) to[bend right=15] (E0);

\node[align=center, font=\footnotesize] at (3,-2.0) {%
\textcolor{red}{deleted edges}\qquad
\textcolor{blue}{added edges}
};

\end{tikzpicture}%
}
\caption{The switching step in the construction of $G_j$. Starting from the disjoint union of the incidence component $I_j$ and the gadget $P_r^j$, we delete the edges $x_jF_j$ and $b_r^j\overline E_0^j$, and add the edges $x_jb_r^j$ and $F_j\overline E_0^j$.}
\label{fig:switchingF}
\end{figure}
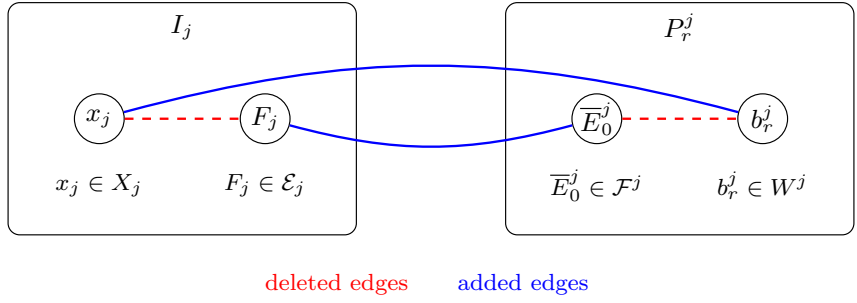

We now present the main theorem of this section. The switching step of its proof is illustrated in Figure~\ref{fig:switchingF}.

\begin{theorem}
\label{thm:np-complete-regular}
Given a bipartite $r$-regular graph $G$, deciding whether
\[
    \widetilde{\gamma}_r(G)=\frac{|V(G)|}{2}
\]
is NP-complete for every fixed $r\ge3$.
\end{theorem}
\begin{proof}
By Corollary~\ref{cor:regular-lower}, $\widetilde{\gamma}_r(G)\ge |V(G)|/2$. Thus, a rainbow $r$-dominating function of weight $|V(G)|/2$ certifies equality. Membership in NP follows by discussion given in Section~\ref{sec:preliminaries}.

Let $\mathcal H=(X,\mathcal E)$ be an instance of \textsc{Strong $r$-Coloring of Uniform Regular Hypergraphs}. Construct its incidence graph $I(\mathcal H)$ by joining each vertex $x\in X$ to the hyperedges containing it. Since $\mathcal H$ is $r$-uniform and $r$-regular, $I(\mathcal H)$ is $r$-regular. Let $I_1,\ldots,I_s$ be the connected components of $I(\mathcal H)$, and let $(X_j,\mathcal E_j)$ be their respective partite sets. Then, counting incidences in each $I_j$ gives $|X_j|=|\mathcal E_j|$.

For each $j\in[s]$, choose an edge $x_jF_j\in E(I_j)$,
where $x_j\in X_j$ and $F_j\in\mathcal E_j$, and add a disjoint copy $P_r^j$ of the forcing gadget introduced before. Delete edges $x_jF_j$ and $b_r^j\overline E_0^j$, and add edges $x_jb_r^j$ and $F_j\overline E_0^j$. Denote the resulting graph by $G_j$, and let $G=\bigsqcup_{j=1}^sG_j$.

The total number of vertices of $G$ is $N=|V(G)|=2|X|+4rs$. Moreover, each $G_j$ is bipartite and $r$-regular. The two partite sets are
\[
    L_j=X_j\cup\mathcal F^j,
    \qquad R_j=\mathcal E_j\cup W^j.
\]
Furthermore, since $|X_j|=|\mathcal E_j|$, the two sets $L_j$ and $R_j$ are also of equal cardinalities.

Recall that both $I_j$ and $P_r^j$ are connected. Since a connected $r$-regular bipartite graph has no bridge for $r\ge2$, neither $x_jF_j$ nor $b_r^j\overline E_0^j$ is a bridge. Therefore, after deleting the two edges, both $I_j$ and $P_r^j$ remain connected. Since the two added edges join them, each $G_j$ is connected.

The construction procedure can be performed in polynomial time for every fixed $r\ge3$. Indeed, each connected component $I_1,\ldots,I_s$ is joined with exactly one gadget of constant size $4r$.

We now prove that $\mathcal H$ admits a strong $r$-coloring if and only if $\widetilde{\gamma}_r(G)=N/2$.

\smallskip
\noindent(\,$\Rightarrow$\,) Suppose that $\mathcal H$ admits a strong $r$-coloring $\chi$. For each $j\in[s]$, let $\psi_j:\mathcal F^j\to[r]$ be a coloring as in the first point of Lemma~\ref{lem:orientation-gadget}, with the colors permuted so that $\psi_j(\overline E_0^j)=\chi(x_j)$. Construct a function $f:V(G)\to\{\emptyset,\{1\},\ldots,\{r\}\}$ such that
\[
f(v)=\begin{cases}
\{\chi(v)\}, & \text{if } v\in X_j \text{ for some } j\in[s],\\
\{\psi_j(v)\}, & \text{if } v\in\mathcal F^j \text{ for some } j\in[s],\\
\emptyset, & \text{otherwise}.
\end{cases}
\]
Fix $j\in[s]$. Since $\chi$ is a strong $r$-coloring of $\mathcal H$, every vertex $F\in\mathcal E_j\setminus\{F_j\}$ sees all $r$ colors among its neighbors in $X_j$. The vertex $F_j$ has the same neighborhood as in $I(\mathcal H)$, except that $x_j$ is replaced by $\overline E_0^j$. Since $\chi(x_j)=\psi_j(\overline E_0^j)$, $F_j$ also sees all $r$ colors.

Similarly, every vertex of $W^j\setminus\{b_r^j\}$ has the same neighborhood as in $P_r^j$, and hence by Lemma~\ref{lem:orientation-gadget} sees all $r$ colors. The vertex $b_r^j$ has the same neighborhood as in $P_r^j$, except that $\overline E_0^j$ is replaced by $x_j$. Recall that $\psi_j(\overline E_0^j)=\chi(x_j)$, and thus $b_r^j$, as well as every other uncolored vertex of $G_j$, sees all $r$ colors in its neighborhood. This proves that $f$ is a rainbow $r$-dominating function of $G$. Its weight is
\[
\begin{aligned}
w(f)
&=\sum_{j=1}^s\bigl(|X_j|+|\mathcal F^j|\bigr)\\
&=|X|+2rs\\
&=\frac{2|X|+4rs}{2}
=\frac{N}{2}.
\end{aligned}
\]
By Corollary~\ref{cor:regular-lower},
\[
\widetilde{\gamma}_r(G)\ge\frac{N}{2},
\]
and therefore
\[
\widetilde{\gamma}_r(G)=\frac{N}{2}.
\]

\smallskip
\noindent(\,$\Leftarrow$\,) Conversely, suppose that $\widetilde{\gamma}_r(G)=N/2$, and $f$ is a minimum rainbow $r$-dominating function of $G$. By Theorem~\ref{thm:extremal-regular}, $(C_f,U_f)$ is a bipartition of $G$. For each $j\in[s]$, $G_j$ is connected, and thus $C_f\cap V(G_j)=R_j$ or $C_f\cap V(G_j)=L_j$. Since the neighborhoods of the vertices in $\mathcal F^j\setminus\{\overline E_0^j\}$ were not changed by the construction, the former is impossible by Lemma~\ref{lem:orientation-gadget}. Thus, $f(w)=\emptyset$ for every vertex $w\in W^j$. Identifying each singleton $f(v)=\{i\}$ with the color $i$, Lemma~\ref{lem:color-transmission} gives $f(x_j)=f(\overline E_0^j)$ for every $j\in[s]$.

Let $\chi: X\to [r]$ be a function such that
\[
\chi(x)=i \qquad\text{if } f(x)=\{i\}.
\]
For each $j\in[s]$, the construction did not change the neighborhoods of the vertices in $\mathcal E_j$, except for $F_j$. Hence, every $F\in\mathcal E_j\setminus\{F_j\}$ contains all $r$ colors under $\chi$. Since $F_j$ is uncolored, its neighbors in $G$ receive all $r$ colors. Moreover, $f(x_j)=f(\overline E_0^j)$, so replacing $\overline E_0^j$ by $x_j$ preserves the set of colors appearing in the neighborhood of $F_j$. Therefore, $\chi$ is a strong $r$-coloring of $\mathcal H$. This completes the proof.
\end{proof}

Consequently, for every fixed $r\ge3$, the decision version of rainbow $r$-domination is NP-complete, and computing $\widetilde{\gamma}_r(G)$ is NP-hard, even for bipartite $r$-regular graphs.

\section{Exact values for selected regular families}
\label{sec:exact-families}

In this section we apply the bounds and extremal characterization from Section~\ref{sec:bounds} to obtain exact values for several standard families of regular graphs. We first analyze connected $2$-regular graphs, that is, cycles. We then consider two cubic families: M\"obius ladders and prism graphs.

\subsection{Cycles}

We denote the cycle of order $n$ by $C_n$, and label its vertices by $V(C_n)=\{v_0,v_1,\dots,v_{n-1}\}$, with indices taken modulo $n$.

For $k=2$, the following proposition restates the cycle formula established by Chellali, Haynes, and Hedetniemi \cite[Proposition~8]{chellali2026dual}, since dual-server domination coincides with rainbow $2$-domination. We include a self-contained proof using the bounds developed earlier.

\begin{proposition}
\label{prop:cycles}
For every cycle $C_n$,
\[
\widetilde{\gamma}_{k}(C_n)=
\begin{cases}
\left\lceil \dfrac{n}{3}\right\rceil, & \text{if } k=1,\\[2mm]
\dfrac{n}{2}, & \text{if } k=2 \text{ and } n\equiv 0 \pmod 4,\\[2mm]
\left\lceil \dfrac{n+1}{2}\right\rceil, & \text{if } k=2 \text{ and } n\not\equiv 0 \pmod 4,\\[2mm]
n, & \text{if } k\ge 3.
\end{cases}
\]
\end{proposition}

Two cycles of different congruence classes with their minimum rainbow $2$-dominating functions are shown in Figure~\ref{fig:cycles-rainbow2}.

\begin{figure}[ht]
\centering
\begin{tikzpicture}[
    vertex/.style={
        circle,
        draw,
        minimum size=4.5mm,
        inner sep=0pt,
        font=\tiny
    },
    c1/.style={vertex,fill=red!25},
    c2/.style={vertex,fill=blue!25},
    empty/.style={vertex,fill=white}
]

\begin{scope}[xshift=-2.2cm]

\node[empty] (a0) at (0,1.35) {};
\node[c1] (a1) at (0.95,0.95) {1};
\node[empty] (a2) at (1.35,0) {};
\node[c2] (a3) at (0.95,-0.95) {2};
\node[empty] (a4) at (0,-1.35) {};
\node[c1] (a5) at (-0.95,-0.95) {1};
\node[empty] (a6) at (-1.35,0) {};
\node[c2] (a7) at (-0.95,0.95) {2};

\draw
(a0)--(a1)--(a2)--(a3)--(a4)--(a5)--(a6)--(a7)--(a0);

\node[font=\scriptsize] at (0,-1.85)
{(a) $C_8$};

\end{scope}

\begin{scope}[xshift=2.2cm]

\node[c1] (b0) at (0,1.35) {1};
\node[empty] (b1) at (0.73,1.14) {};
\node[c2] (b2) at (1.23,0.57) {2};
\node[empty] (b3) at (1.34,-0.19) {};
\node[c1] (b4) at (1.03,-0.88) {1};
\node[empty] (b5) at (0.38,-1.30) {};
\node[c2] (b6) at (-0.38,-1.30) {2};
\node[empty] (b7) at (-1.03,-0.88) {};
\node[c1] (b8) at (-1.34,-0.19) {1};
\node[empty] (b9) at (-1.23,0.57) {};
\node[c2] (b10) at (-0.73,1.14) {2};

\draw
(b0)--(b1)--(b2)--(b3)--(b4)--(b5)--(b6)--(b7)--(b8)--(b9)--(b10)--(b0);

\node[font=\scriptsize] at (0,-1.85)
{(b) $C_{11}$};

\end{scope}

\end{tikzpicture}
\caption{Cycles $C_8$ and $C_{11}$ with their minimum rainbow $2$-dominating functions. Numbers inside colored vertices indicate the assigned colors, and white vertices are uncolored.}
\label{fig:cycles-rainbow2}
\end{figure}
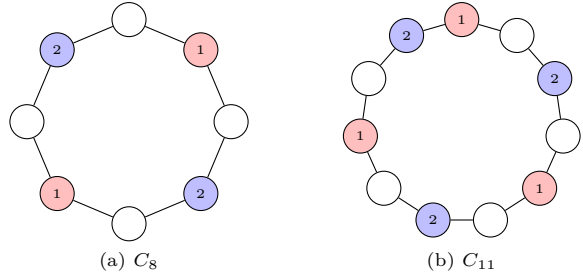

\begin{proof}
For $k=1$, the statement is the classical domination number of a cycle, see
\cite{haynes1998fundamentals}:
\[
\widetilde{\gamma}_{1}(C_n)=\gamma(C_n)=\left\lceil \frac{n}{3}\right\rceil.
\]

For $k=2$, let $q$ and $s$ be integers such that $n=4q+s$ and $0\le s<4$. Since $C_n$ is $2$-regular, Corollary~\ref{cor:regular-lower} gives
\[
\widetilde{\gamma}_{2}(C_n) \ge 2q+\min\{2,s\}.
\]
Equivalently,
\[
\widetilde{\gamma}_{2}(C_n)\ge
\begin{cases}
\dfrac{n}{2}, & \text{if } n\equiv0\pmod4,\\[2mm]
\left\lceil\dfrac{n+1}{2}\right\rceil, & \text{if } n\not\equiv0\pmod4.
\end{cases}
\]

We now prove the corresponding upper bounds by constructing rainbow $2$-dominating functions attaining them.

First, suppose that $n\equiv0\pmod4$. Let $f:V(C_n)\to\{\emptyset, \{1\},\{2\}\}$ be a function such that
\[
f(v_i)=
\begin{cases}
\{1\}, & \text{if } i\equiv 1 \pmod 4,\\
\{2\}, & \text{if } i\equiv 3 \pmod 4,\\
\emptyset, & \text{if } i \text{ is even.}
\end{cases}
\]
Since every uncolored vertex $v_i$ has neighbors $v_{i-1}$ and $v_{i+1}$, and those neighbors are assigned colors $\{1\}$ and $\{2\}$ in some order, $f$ is a rainbow $2$-dominating function. Its weight is exactly $n/2$. Hence
\[
\widetilde{\gamma}_{2}(C_n)\le\frac{n}{2}.
\]

Suppose next that $n\equiv2\pmod4$. Define
\[
f(v_i)=
\begin{cases}
\{1\}, & \text{if } i\equiv 1 \pmod 4,\\
\{2\}, & \text{if } i\equiv 3 \pmod 4 \text{ or } i=0,\\
\emptyset, & \text{otherwise.}
\end{cases}
\]
Then $f$ is a rainbow $2$-dominating function of $C_n$. Indeed, every uncolored vertex $v_i$, where $i\in\{2,4,\dots,n-2\}$, has neighbors colored with $1$ and $2$. Since exactly $n/2+1$ vertices are colored by $f$, it follows that
\[
\widetilde{\gamma}_{2}(C_n)\le\frac{n}{2}+1
=\left\lceil\frac{n+1}{2}\right\rceil.
\]

Lastly, suppose that $n$ is odd. Define
\[
f(v_i)=
\begin{cases}
\{1\}, & \text{if } i\equiv 0 \pmod 4,\\
\{2\}, & \text{if } i\equiv 2 \pmod 4,\\
\emptyset, & \text{if } i \text{ is odd.}
\end{cases}
\]
Then every uncolored vertex $v_i$ has two even-indexed neighbors, and those neighbors are assigned two distinct colors. The weight of $f$ is exactly $(n+1)/2$, and therefore
\[
\widetilde{\gamma}_{2}(C_n)\le\frac{n+1}{2}
=\left\lceil\frac{n+1}{2}\right\rceil.
\]

Combining the established lower and upper bounds gives the result for $k=2$.

Finally, if $k\ge 3$, then no vertex of $C_n$ can remain uncolored. Indeed, every vertex has degree $2$ and therefore cannot see more than two distinct colors in its neighborhood. Hence
\[
\widetilde{\gamma}_{k}(C_n)=n\qquad \text{if } k\geq 3.
\]
This completes the proof.
\end{proof}

Since every connected $2$-regular graph is a cycle, Proposition~\ref{prop:cycles} determines the rainbow $k$-domination number for all connected $2$-regular graphs.

\subsection{M\"obius ladders}

The M\"obius ladder $M_n$, where $n\ge 4$ is even, is a cubic graph obtained from the cycle
\[
    v_0v_1\cdots v_{n-1}v_0
\]
by adding the edges
\[
    v_i v_{i+n/2}
\]
for all $i\in\{0,1,\ldots,n/2-1\}$, with indices taken modulo $n$.

Kuzman obtained a formula for the classical $3$-rainbow domination number of M\"obius ladders \cite[Theorem~4.5]{kuzman2020}. For $n\equiv0,4,6\pmod{12}$, his upper-bound constructions use only empty and one-element sets, so the stated values also hold for rainbow $3$-domination. We determine the exact rainbow $3$-domination number of this family for every relevant congruence class of $n$.

\begin{theorem}
\label{thm:mobius}
Let $n\ge 4$ be even, and let $M_n$ denote the M\"obius ladder on $n$ vertices. Then
\[
\widetilde{\gamma}_{3}(M_n)=
\begin{cases}
\dfrac{n}{2}, & \text{if } n \equiv 6 \pmod{12},\\[2mm]
\dfrac{n}{2}+1, & \text{if } n \equiv 0,4,8 \pmod{12},\\[2mm]
\dfrac{n}{2}+2, & \text{if } n \equiv 2,10 \pmod{12}.
\end{cases}
\]
\end{theorem}

The proof of Theorem~\ref{thm:mobius} is split into several lemmas. Their proofs are case-dependent modular constructions and counting arguments. They are given in Appendix~\ref{app:mobius-proofs}. Examples of M\"obius ladders with their minimum rainbow $3$-dominating functions are shown in Figure~\ref{fig:mobius-rainbow3}.

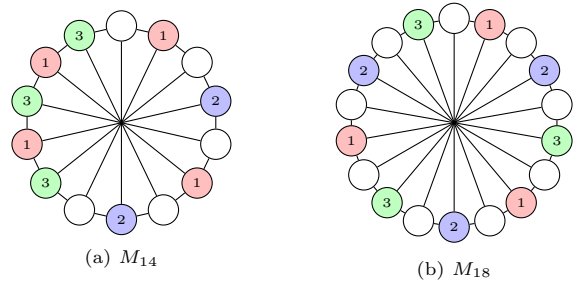
\begin{figure}[ht]
\centering
\begin{tikzpicture}[
    vertex/.style={
        circle,
        draw,
        minimum size=4.0mm,
        inner sep=0pt,
        font=\tiny
    },
    c1/.style={vertex,fill=red!25},
    c2/.style={vertex,fill=blue!25},
    c3/.style={vertex,fill=green!25},
    empty/.style={vertex,fill=white}
]

\begin{scope}[xshift=-2.2cm]

\foreach \i/\lab/\sty in {
0/{}/empty,
1/1/c1,
2/{}/empty,
3/2/c2,
4/{}/empty,
5/1/c1,
6/{}/empty,
7/2/c2,
8/{}/empty,
9/3/c3,
10/1/c1,
11/3/c3,
12/1/c1,
13/3/c3}
{
    \node[\sty] (a\i) at ({90-\i*360/14}:1.28) {\lab};
}

\draw
(a0)--(a1)--(a2)--(a3)--(a4)--(a5)--(a6)--(a7)--
(a8)--(a9)--(a10)--(a11)--(a12)--(a13)--(a0);

\draw (a0)--(a7);
\draw (a1)--(a8);
\draw (a2)--(a9);
\draw (a3)--(a10);
\draw (a4)--(a11);
\draw (a5)--(a12);
\draw (a6)--(a13);

\node[font=\scriptsize] at (0,-1.8) {(a) $M_{14}$};

\end{scope}

\begin{scope}[xshift=2.2cm]

\foreach \i/\lab/\sty in {
0/{}/empty,
1/1/c1,
2/{}/empty,
3/2/c2,
4/{}/empty,
5/3/c3,
6/{}/empty,
7/1/c1,
8/{}/empty,
9/2/c2,
10/{}/empty,
11/3/c3,
12/{}/empty,
13/1/c1,
14/{}/empty,
15/2/c2,
16/{}/empty,
17/3/c3}
{
    \node[\sty] (b\i) at ({90-\i*360/18}:1.38) {\lab};
}

\draw
(b0)--(b1)--(b2)--(b3)--(b4)--(b5)--(b6)--(b7)--(b8)--
(b9)--(b10)--(b11)--(b12)--(b13)--(b14)--(b15)--(b16)--(b17)--(b0);

\draw (b0)--(b9);
\draw (b1)--(b10);
\draw (b2)--(b11);
\draw (b3)--(b12);
\draw (b4)--(b13);
\draw (b5)--(b14);
\draw (b6)--(b15);
\draw (b7)--(b16);
\draw (b8)--(b17);

\node[font=\scriptsize] at (0,-1.95) {(b) $M_{18}$};

\end{scope}

\end{tikzpicture}
\caption{M\"obius ladders $M_{14}$ and $M_{18}$ with their minimum rainbow $3$-dominating functions. Numbers inside colored vertices indicate the assigned colors, and white vertices are uncolored.}
\label{fig:mobius-rainbow3}
\end{figure}

\begin{lemma}
\label{lem:mobius-upper-6}
If $n\equiv 6 \pmod{12}$, then
\[
\widetilde{\gamma}_{3}(M_n)\le \frac{n}{2}.
\]
\end{lemma}

\begin{lemma}
\label{lem:mobius-plus-two-lower}
If $n\equiv 2,10 \pmod{12}$, then
\[
\widetilde{\gamma}_{3}(M_n)\ge \frac{n}{2}+2.
\]
\end{lemma}

\begin{lemma}
\label{lem:mobius-upper-0mod4}
If $n\equiv 0\pmod 4$, then
\[
\widetilde{\gamma}_{3}(M_n)\le \frac{n}{2}+1.
\]
\end{lemma}

\begin{lemma}
\label{lem:mobius-upper-2mod12}
If $n\equiv 2\pmod{12}$, then
\[
\widetilde{\gamma}_{3}(M_n)\le \frac{n}{2}+2.
\]
\end{lemma}

\begin{lemma}
\label{lem:mobius-upper-10mod12}
If $n\equiv 10\pmod{12}$, then
\[
\widetilde{\gamma}_{3}(M_n)\le \frac{n}{2}+2.
\]
\end{lemma}

\begin{proof}[Proof of Theorem~\ref{thm:mobius}]
Since $M_n$ is cubic, Corollary~\ref{cor:cubic-extremal} gives
\[
\widetilde{\gamma}_{3}(M_n)\ge
\begin{cases}
\dfrac{n}{2}, & \text{if } n\equiv0,6\pmod{12},\\[2mm]
\dfrac{n}{2}+1, & \text{if } n\equiv2,4,8,10\pmod{12}.
\end{cases}
\]

First, suppose that $n\equiv6\pmod{12}$. Then Lemma~\ref{lem:mobius-upper-6} gives
\[
\widetilde{\gamma}_{3}(M_n)\le\frac{n}{2},
\]
and therefore
\[
\widetilde{\gamma}_{3}(M_n)=\frac{n}{2}.
\]

Now suppose that $n\equiv0,4,8\pmod{12}$. If $n\equiv4,8\pmod{12}$, Corollary~\ref{cor:cubic-extremal} gives
\[
\widetilde{\gamma}_{3}(M_n)\ge\frac{n}{2}+1.
\]
If $n\equiv0\pmod{12}$, then $M_n$ is not bipartite. Indeed, in this case $n/2$ is even and every edge $v_iv_{i+n/2}$ joins vertices of the same parity. Consequently, by Theorem~\ref{thm:extremal-regular},
\[
\widetilde{\gamma}_{3}(M_n)\ge\frac{n}{2}+1.
\]
Since Lemma~\ref{lem:mobius-upper-0mod4} gives the corresponding upper bound, it follows that
\[
\widetilde{\gamma}_{3}(M_n)=\frac{n}{2}+1.
\]

Lastly, suppose that $n\equiv 2,10\pmod{12}$. Then Lemma~\ref{lem:mobius-plus-two-lower} gives
\[
\widetilde{\gamma}_{3}(M_n)\ge\frac{n}{2}+2.
\]
Combining this with Lemma~\ref{lem:mobius-upper-2mod12} for $n\equiv2\pmod{12}$ and Lemma~\ref{lem:mobius-upper-10mod12} for $n\equiv10\pmod{12}$, we obtain
\[
\widetilde{\gamma}_{3}(M_n)=\frac{n}{2}+2.
\]

This completes the proof.
\end{proof}

\subsection{Prism graphs}

The prism graph $Pr(m)$ is the Cartesian product
\[
    Pr(m)=C_m\square K_2.
\]
Equivalently, it has vertex set
\[
    \{u_0,u_1,\ldots,u_{m-1}\}\cup
    \{v_0,v_1,\ldots,v_{m-1}\},
\]
and edge set consisting of the two cycles
\[
    u_0u_1\cdots u_{m-1}u_0,\qquad
    v_0v_1\cdots v_{m-1}v_0,
\]
together with the matching edges $u_i v_i$ for $i\in\{0,1,\ldots,m-1\}$. Again, indices are taken modulo $m$.

We now give the exact value of the rainbow $k$-domination number of prism graphs.

\begin{theorem}
\label{thm:prisms}
Let $m\ge 3$. Then, for every integer $k\ge 1$,
\[
\widetilde{\gamma}_{k}(Pr(m))=
\begin{cases}
\dfrac{m}{2}, & \text{if } k=1 \text{ and } m\equiv 0 \pmod 4,\\[2mm]
\dfrac{m+1}{2}, & \text{if } k=1 \text{ and } m\equiv 1,3 \pmod 4,\\[2mm]
\dfrac{m}{2}+1, & \text{if } k=1 \text{ and } m\equiv 2 \pmod 4,\\[2mm]
m, & \text{if } k=2,\\[1mm]
m, & \text{if } k=3 \text{ and } m\equiv 0 \pmod 6,\\[1mm]
m+1, & \text{if } k=3 \text{ and } m\equiv 1,3,5 \pmod 6,\\[1mm]
m+2, & \text{if } k=3 \text{ and } m\equiv 2,4 \pmod 6,\\[1mm]
2m, & \text{if } k\ge 4.
\end{cases}
\]
\end{theorem}

The proof of Theorem~\ref{thm:prisms} utilizes the following two lemmas, whose proofs are given in Appendix~\ref{app:prism-proofs}, as well as known results for $k$-rainbow domination and rainbow $k$-domination numbers. Figure~\ref{fig:prism-rainbow} illustrates an exemplary prism graph with its minimum rainbow $2$- and $3$-dominating functions.

\begin{figure}[ht]
\centering
\begin{tikzpicture}[
    vertex/.style={
        circle,
        draw,
        minimum size=4.0mm,
        inner sep=0pt,
        font=\tiny
    },
    c1/.style={vertex,fill=red!25},
    c2/.style={vertex,fill=blue!25},
    c3/.style={vertex,fill=green!25},
    empty/.style={vertex,fill=white}
]

\begin{scope}[xshift=-2.2cm]

\node[c1] (au0) at (0,1.35) {1};
\node[c2] (au1) at (0.95,0.95) {2};
\node[empty] (au2) at (1.35,0) {};
\node[c2] (au3) at (0.95,-0.95) {2};
\node[c2] (au4) at (0,-1.35) {2};
\node[empty] (au5) at (-0.95,-0.95) {};
\node[c1] (au6) at (-1.35,0) {1};
\node[empty] (au7) at (-0.95,0.95) {};

\node[empty] (av0) at (0,0.68) {};
\node[empty] (av1) at (0.48,0.48) {};
\node[c1] (av2) at (0.68,0) {1};
\node[empty] (av3) at (0.48,-0.48) {};
\node[empty] (av4) at (0,-0.68) {};
\node[c1] (av5) at (-0.48,-0.48) {1};
\node[empty] (av6) at (-0.68,0) {};
\node[c2] (av7) at (-0.48,0.48) {2};

\draw
(au0)--(au1)--(au2)--(au3)--(au4)--(au5)--(au6)--(au7)--(au0);

\draw
(av0)--(av1)--(av2)--(av3)--(av4)--(av5)--(av6)--(av7)--(av0);

\draw
(au0)--(av0)
(au1)--(av1)
(au2)--(av2)
(au3)--(av3)
(au4)--(av4)
(au5)--(av5)
(au6)--(av6)
(au7)--(av7);

\node[font=\scriptsize] at (0,-1.85)
{(a) $\widetilde{\gamma}_2(Pr(8))=8$};

\end{scope}

\begin{scope}[xshift=2.2cm]

\node[empty] (bu0) at (0,1.35) {};
\node[c1] (bu1) at (0.95,0.95) {1};
\node[empty] (bu2) at (1.35,0) {};
\node[c2] (bu3) at (0.95,-0.95) {2};
\node[empty] (bu4) at (0,-1.35) {};
\node[c1] (bu5) at (-0.95,-0.95) {1};
\node[empty] (bu6) at (-1.35,0) {};
\node[c3] (bu7) at (-0.95,0.95) {3};

\node[c2] (bv0) at (0,0.68) {2};
\node[empty] (bv1) at (0.48,0.48) {};
\node[c3] (bv2) at (0.68,0) {3};
\node[c1] (bv3) at (0.48,-0.48) {1};
\node[c3] (bv4) at (0,-0.68) {3};
\node[empty] (bv5) at (-0.48,-0.48) {};
\node[c2] (bv6) at (-0.68,0) {2};
\node[c1] (bv7) at (-0.48,0.48) {1};

\draw
(bu0)--(bu1)--(bu2)--(bu3)--(bu4)--(bu5)--(bu6)--(bu7)--(bu0);

\draw
(bv0)--(bv1)--(bv2)--(bv3)--(bv4)--(bv5)--(bv6)--(bv7)--(bv0);

\draw
(bu0)--(bv0)
(bu1)--(bv1)
(bu2)--(bv2)
(bu3)--(bv3)
(bu4)--(bv4)
(bu5)--(bv5)
(bu6)--(bv6)
(bu7)--(bv7);

\node[font=\scriptsize] at (0,-1.85)
{(b) $\widetilde{\gamma}_3(Pr(8))=10$};

\end{scope}

\end{tikzpicture}

\caption{The prism graph $Pr(8)$ with its minimum rainbow $2$- and $3$-dominating functions. Numbers inside colored vertices indicate the assigned colors, and white vertices are uncolored.}
\label{fig:prism-rainbow}
\end{figure}
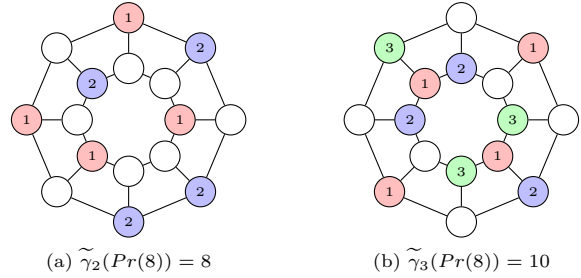

\begin{lemma}
\label{lem:prism-upper-2}
For every $m\geq3$,
\[
\widetilde{\gamma}_{2}(Pr(m))\leq m.
\]
\end{lemma}

\begin{lemma}
\label{lem:prism-lower-3}
If $m\equiv 2\pmod6$, then
\[
\widetilde{\gamma}_{3}(Pr(m))\geq m+2.
\]
\end{lemma}

\begin{proof}[Proof of Theorem~\ref{thm:prisms}]
For $k=1$, rainbow $1$-domination is equivalent to ordinary domination and $1$-rainbow domination. The stated values are already known in the literature \cite[Theorem~4.2]{kuzman2020}.

For $k=2$, Kuzman \cite[Theorem~4.2]{kuzman2020} proved that
\[
\gamma_{r2}(Pr(m))=m.
\]
Hence, by Proposition~\ref{prop:rainbow-relation},
\[
\widetilde{\gamma}_{2}(Pr(m))\ge m.
\]
By combining this with the corresponding upper bound given in Lemma~\ref{lem:prism-upper-2}, we obtain
\[
\widetilde{\gamma}_{2}(Pr(m))=m.
\]

For $k=3$, the classical values in \cite[Theorem~4.2]{kuzman2020}, together with Proposition~\ref{prop:rainbow-relation}, give the required lower bounds except when $m\equiv2\pmod6$. In that case, the required lower bound follows from Lemma~\ref{lem:prism-lower-3}.

If $m=3$, let $f$ be a function such that
\[
f(u_1)=\{1\},\qquad f(v_0)=\{2\},\qquad f(u_2)=f(v_2)=\{3\},
\]
and
\[
f(u_0)=f(v_1)=\emptyset.
\]
A direct check shows that $f$ is a rainbow $3$-dominating function of weight $4=m+1$.

If $m\geq 4$, apply \v{Z}erovnik's result \cite[Theorem~7]{zerovnik2026singleton} to the cylindrical graph $P_2\square C_m=Pr(m)$. This gives exactly the required upper bounds.

Finally, if $k\ge 4$, then no vertex of $Pr(m)$ can remain uncolored. Indeed, since $Pr(m)$ is cubic, no vertex can see $k$ distinct colors in its neighborhood. Hence
\[
\widetilde{\gamma}_{k}(Pr(m))=|V(Pr(m))|=2m.
\]

This completes the proof.
\end{proof}

\section{Dynamic programming for generalized Petersen graphs}
\label{sec:dp}

We now turn to generalized Petersen graphs of the form $P(6t,t)$. This family provides another natural cubic setting for rainbow $k$-domination and admits a finite-state dynamic-programming algorithm.

The generalized Petersen graph $P(6t,t)$ has vertex set
\[
V(P(6t,t))=
    \{u_0,u_1,\ldots,u_{6t-1}\}\cup
    \{v_0,v_1,\ldots,v_{6t-1}\},
\]
and edge set
\[
\begin{aligned}
E(P(6t,t))={}&
    \{u_i u_{i+1}:0\le i\le 6t-1\} \\
&\cup \{u_i v_i:0\le i\le 6t-1\} \\
&\cup \{v_i v_{i+t}:0\le i\le 6t-1\},
\end{aligned}
\]
with indices taken modulo $6t$. Hence $|V(P(6t,t))|=12t$. In this section we assume that $t\geq2$.

Rainbow-type domination parameters in generalized Petersen graphs have been studied in several earlier works. The case of $2$-rainbow domination was considered by Shao et al.~\cite{shao2019two}, and later more systematically for the family $P(ct,t)$ by Brezovnik, Rupnik Poklukar, and \v{Z}erovnik~\cite{brezovnik2023two}. The case of $3$-rainbow domination for $P(6t,t)$ was studied by Erve\v{s} and \v{Z}erovnik~\cite{erves2021petersen}, and broader families $P(ct,t)$ were considered by Rupnik Poklukar and \v{Z}erovnik~\cite{rupnik2022three}. These works provide the main background for the computational results and conjectures formulated below.

Since $P(6t,t)$ is cubic, every uncolored vertex can see at most three distinct colors in its neighborhood. Consequently,
\[
    \widetilde{\gamma}_{k}(P(6t,t))=|V(P(6t,t))|=12t
    \qquad \text{for every } k\ge 4.
\]
Therefore, only the cases $k\in\{1,2,3\}$ are nontrivial. In this section we present an exact algorithm for fixed $k$ and prove that it runs in linear time with respect to $t$.

For $j\in\{0,1,\ldots,t-1\}$ and $r\in\{0,1,\ldots,5\}$, define
\[
    u_j^r=u_{j+rt},
    \qquad
    v_j^r=v_{j+rt},
\]
where the superscript is taken modulo $6$. Then the vertex set of $P(6t,t)$ is partitioned into the $t$ blocks
\[
    B_j=\{u_j^r,v_j^r:r=0,1,\ldots,5\},
    \qquad
    j=0,1,\ldots,t-1.
\]
The structure of one such block was illustrated in Figure~\ref{fig:petersen-block}.

\begin{figure}[ht]
\centering
\resizebox{0.5\columnwidth}{!}{%
\begin{tikzpicture}[
    vtx/.style={
        circle,
        draw=black,
        fill=white,
        minimum size=4.5mm,
        inner sep=0pt,
        line width=0.5pt
    },
    blockvtx/.style={
        circle,
        draw=red,
        fill=white,
        minimum size=4.5mm,
        inner sep=0pt,
        line width=0.9pt
    },
    edge/.style={
        draw=black,
        line width=0.5pt
    },
    blockedge/.style={
        draw=red,
        line width=0.9pt
    }
]

\foreach \i in {0,...,17} {
    \pgfmathsetmacro{\ang}{90-20*\i}
    \node[vtx] (u\i) at (\ang:3.8) {};
    \node[vtx] (v\i) at (\ang:2.8) {};
}

\draw[edge]
(u0)--(u1)--(u2)--(u3)--(u4)--(u5)--(u6)--(u7)--(u8)--(u9)--(u10)--(u11)--(u12)--(u13)--(u14)--(u15)--(u16)--(u17)--(u0);

\foreach \i in {0,...,17} {
    \draw[edge] (u\i)--(v\i);
}

\draw[edge] (v0)--(v3)--(v6)--(v9)--(v12)--(v15)--(v0);
\draw[edge] (v1)--(v4)--(v7)--(v10)--(v13)--(v16)--(v1);
\draw[edge] (v2)--(v5)--(v8)--(v11)--(v14)--(v17)--(v2);

\draw[blockedge] (v0)--(v3)--(v6)--(v9)--(v12)--(v15)--(v0);

\foreach \i in {0,3,6,9,12,15} {
    \draw[blockedge] (u\i)--(v\i);
}

\foreach \i in {0,3,6,9,12,15} {
    \node[blockvtx] at (u\i) {};
    \node[blockvtx] at (v\i) {};
}

\end{tikzpicture}%
}
\caption{The generalized Petersen graph $P(18,3)$ with the block $B_0$ highlighted in red.}
\label{fig:petersen-block}
\end{figure}
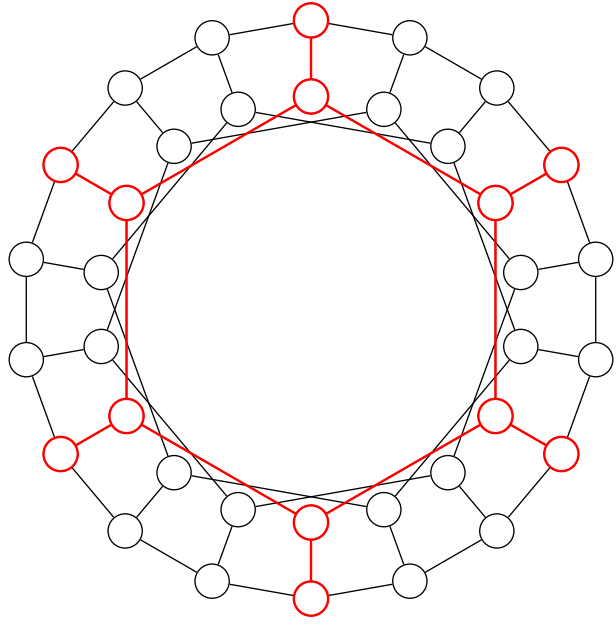

In this representation, for each fixed $j$ the vertices
\[
    v_j^0,v_j^1,\ldots,v_j^5
\]
form a $6$-cycle, each vertex $u_j^r$ is adjacent to $v_j^r$, and the outer edges are
\[
    u_j^r u_{j+1}^r
    \qquad
    (0\le j\le t-2),
\]
together with the wrap-around edges
\[
    u_{t-1}^r u_0^{r+1}.
\]
Thus, in the context of rainbow $k$-domination, each block interacts directly only with its two neighboring blocks. Because this interaction has constant width, an algorithm of computational complexity linear in the number of blocks can be constructed.

Additionally, for $0\le j\le t-2$, the blocks $B_0\cup B_j\cup B_{j+1}$ form a path decomposition of width at most $35$. Consequently, linear-time solvability for fixed $k$ can also be derived by adapting the bounded-treewidth dynamic program for classical $k$-rainbow domination constructed by Meybodi et al.~\cite{meybodi2021treewidth} to forbid assignments of multi-element sets. Our algorithm was developed independently and explicitly utilizes rainbow $k$-domination characteristics as well as specific structural features of $P(6t,t)$.

\subsection{A linear-time algorithm}

Throughout this subsection we assume that $1\le k\le 3$.

Let
\[
    \Sigma_k=\{0,1,\ldots,k\}^6.
\]
An element
\[
    \sigma=(\sigma_0,\sigma_1,\ldots,\sigma_5)\in\Sigma_k
\]
will be called an outer state. We interpret $\sigma_r=0$ as the vertex $u_j^r$ being uncolored, and $\sigma_r=i$ as $u_j^r$ being assigned color $\{i\}$. Analogously, an element of $\Sigma_k$ will also be used as an inner state for the vertices $v_j^0,\ldots,v_j^5$.

For a state $\sigma\in\Sigma_k$, let
\[
    c(\sigma)=|\{r\in\{0,\ldots,5\}:\sigma_r\neq 0\}|
\]
denote the number of colored vertices represented by $\sigma$.

We also define two shift operators on $\Sigma_k$ by
\[
    \rho^+(\sigma)_r=\sigma_{r+1},
    \qquad
    \rho^-(\sigma)_r=\sigma_{r-1},
\]
where the indices are taken modulo $6$. These shifts are used to express the wrap-around edges joining the last block with the first one.

For three outer states $\alpha,\beta,\gamma\in\Sigma_k$, define $\mu_k(\alpha,\beta,\gamma)$ to be the minimum value of $c(\tau)$ over all inner states
\[
    \tau=(\tau_0,\ldots,\tau_5)\in\Sigma_k
\]
satisfying the following two conditions for every $r\in\{0,\ldots,5\}$.

\begin{enumerate}
    \item If $\beta_r=0$, then
    \[
        \{\alpha_r,\tau_r,\gamma_r\}\setminus\{0\}=[k].
    \]
    \item If $\tau_r=0$, then
    \[
        \{\tau_{r-1},\beta_r,\tau_{r+1}\}\setminus\{0\}=[k].
    \]
\end{enumerate}

If no such inner state exists, we set
\[
    \mu_k(\alpha,\beta,\gamma)=\infty.
\]

Condition (1) expresses that an uncolored outer vertex must see all $k$ colors among its three neighbors: one in the previous block, one in the next block, and one assigned to its incident inner vertex. Condition (2) expresses the analogous requirement for an uncolored inner vertex, whose neighbors are the two adjacent inner vertices in the $6$-cycle and the corresponding outer vertex.

The next lemma gives an exact reformulation of the problem in terms of outer states.

\begin{lemma}
\label{lem:petersen-state-formula}
For every $t\ge 2$ and every $1\le k\le 3$,
\[
\widetilde{\gamma}_{k}(P(6t,t))
=
\min_{\sigma_0,\ldots,\sigma_{t-1}\in\Sigma_k}
\Bigg(
\sum_{j=0}^{t-1} c(\sigma_j)
+\mu_k(\rho^-(\sigma_{t-1}),\sigma_0,\sigma_1)
\]
\[
+\sum_{j=1}^{t-2}\mu_k(\sigma_{j-1},\sigma_j,\sigma_{j+1})
+\mu_k(\sigma_{t-2},\sigma_{t-1},\rho^+(\sigma_0))
\Bigg).
\]
\end{lemma}

\begin{proof}
Let $f$ be a rainbow $k$-dominating function of $P(6t,t)$. For each
$j\in\{0,\ldots,t-1\}$, let
\[
    \sigma_j=(\sigma_{j,0},\ldots,\sigma_{j,5})\in\Sigma_k
\]
be the outer state induced by $f$ on the vertices $u_j^0,\ldots,u_j^5$, and let
\[
    \tau_j=(\tau_{j,0},\ldots,\tau_{j,5})\in\Sigma_k
\]
be the inner state induced by $f$ on the vertices $v_j^0,\ldots,v_j^5$.

For each block $B_j$, the rainbow $k$-domination constraints on the vertices of $B_j$ involve only the outer states of the neighboring blocks and the inner state of $B_j$ itself. Equivalently, for $1\le j\le t-2$, the state $\tau_j$ is feasible for the triple
\[
    (\sigma_{j-1},\sigma_j,\sigma_{j+1}),
\]
for $j=0$ it is feasible for
\[
    (\rho^-(\sigma_{t-1}),\sigma_0,\sigma_1),
\]
and for $j=t-1$ it is feasible for
\[
    (\sigma_{t-2},\sigma_{t-1},\rho^+(\sigma_0)).
\]
Hence
\[
    c(\tau_0)\ge \mu_k(\rho^-(\sigma_{t-1}),\sigma_0,\sigma_1),
\]
\[
    c(\tau_j)\ge \mu_k(\sigma_{j-1},\sigma_j,\sigma_{j+1})
    \qquad \text{if }1\le j\le t-2,
\]
and
\[
    c(\tau_{t-1})\ge
    \mu_k(\sigma_{t-2},\sigma_{t-1},\rho^+(\sigma_0)).
\]
Summing over all blocks yields
\[
\begin{aligned}
w(f)\ge
&\sum_{j=0}^{t-1} c(\sigma_j)
+\mu_k(\rho^-(\sigma_{t-1}),\sigma_0,\sigma_1) \\
&+\sum_{j=1}^{t-2}\mu_k(\sigma_{j-1},\sigma_j,\sigma_{j+1})
+\mu_k(\sigma_{t-2},\sigma_{t-1},\rho^+(\sigma_0)).
\end{aligned}
\]

Conversely, suppose that the outer states $\sigma_0,\ldots,\sigma_{t-1}$ are given, and that their corresponding terms $\mu_k$ are finite. For each block, choose an inner state realizing the corresponding minimum $\mu_k$. Then, by definition of $\mu_k$, every uncolored vertex of every block sees all $k$ colors in its neighborhood. Hence these outer and inner states together yield a rainbow $k$-dominating function of $P(6t,t)$, and its weight is exactly the expression shown. Therefore, the minimum of that expression over all choices of outer states is equal to $\widetilde{\gamma}_{k}(P(6t,t))$.
\end{proof}

Lemma~\ref{lem:petersen-state-formula} leads to a dynamic-programming algorithm over all consecutive blocks. For fixed $k$, Algorithm~\ref{alg:local-cost-table} precomputes the local cost table $\mu_k$ in time constant in $t$. Then, Algorithm~\ref{alg:petersen-dp} performs a dynamic program over the sequence of blocks, and its computational complexity is linear in the number of blocks, or equivalently linear in $t$.

\begin{algorithm}[t]
\caption{Precomputation of the local cost table $\mu_k$}
\label{alg:local-cost-table}
\begin{algorithmic}[1]
\REQUIRE An integer $k\in\{1,2,3\}$ and the state set
$\Sigma_k=\{0,1,\ldots,k\}^6$
\ENSURE The table $\mu_k(\alpha,\beta,\gamma)$ for all
$\alpha,\beta,\gamma\in\Sigma_k$
\FORALL{$\sigma\in\Sigma_k$}
    \STATE compute $c(\sigma)$
\ENDFOR
\FORALL{$(\alpha,\beta,\gamma)\in\Sigma_k^3$}
    \STATE $\mu_k(\alpha,\beta,\gamma)\gets\infty$
\ENDFOR
\FORALL{$(\alpha,\beta,\gamma)\in\Sigma_k^3$}
    \FORALL{$\tau=(\tau_0,\ldots,\tau_5)\in\Sigma_k$}
        \STATE $\mathrm{feasible}\gets \mathrm{true}$
        \FOR{$r\gets 0$ to $5$}
            \IF{$\beta_r=0$ and
            $\{\alpha_r,\tau_r,\gamma_r\}\setminus\{0\}\neq [k]$}
                \STATE $\mathrm{feasible}\gets \mathrm{false}$
            \ENDIF
            \IF{$\tau_r=0$ and
            $\{\tau_{r-1},\beta_r,\tau_{r+1}\}\setminus\{0\}\neq [k]$}
                \STATE $\mathrm{feasible}\gets \mathrm{false}$
            \ENDIF
        \ENDFOR
        \IF{$\mathrm{feasible}$}
            \STATE $\mu_k(\alpha,\beta,\gamma)\gets
            \min\{\mu_k(\alpha,\beta,\gamma),c(\tau)\}$
        \ENDIF
    \ENDFOR
\ENDFOR
\end{algorithmic}
\end{algorithm}

\begin{algorithm}[t]
\caption{Exact dynamic program for $\widetilde{\gamma}_{k}(P(6t,t))$}
\label{alg:petersen-dp}
\begin{algorithmic}[1]
\REQUIRE An integer $t\ge 2$, an integer $k\in\{1,2,3\}$, the state set
$\Sigma_k$, the local cost table $\mu_k$, and the shift operators
$\rho^+,\rho^-$
\ENSURE $\widetilde{\gamma}_{k}(P(6t,t))$
\STATE $\mathrm{best}\gets\infty$
\FORALL{$(\sigma_0,\sigma_1)\in\Sigma_k^2$}
    \FORALL{$(\alpha,\beta)\in\Sigma_k^2$}
        \STATE $D_1(\alpha,\beta)\gets\infty$
    \ENDFOR
    \STATE $D_1(\sigma_0,\sigma_1)\gets c(\sigma_0)+c(\sigma_1)$
    \FOR{$j\gets 1$ to $t-2$}
        \FORALL{$(\beta,\gamma)\in\Sigma_k^2$}
            \STATE $D_{j+1}(\beta,\gamma)\gets\infty$
        \ENDFOR
        \FORALL{$(\alpha,\beta,\gamma)\in\Sigma_k^3$}
            \STATE $D_{j+1}(\beta,\gamma)\gets
            \min\{D_{j+1}(\beta,\gamma),
            D_j(\alpha,\beta)+c(\gamma)+\mu_k(\alpha,\beta,\gamma)\}$
        \ENDFOR
    \ENDFOR
    \FORALL{$(\alpha,\beta)\in\Sigma_k^2$}
        \STATE $\mathrm{best}\gets
        \min\{\mathrm{best},
        D_{t-1}(\alpha,\beta)
        +\mu_k(\alpha,\beta,\rho^+(\sigma_0))
        +\mu_k(\rho^-(\beta),\sigma_0,\sigma_1)\}$
    \ENDFOR
\ENDFOR
\RETURN $\mathrm{best}$
\end{algorithmic}
\end{algorithm}

\begin{theorem}
\label{thm:petersen-dp}
For every integer $t\ge 2$ and every $1\le k\le 3$, Algorithm~\ref{alg:petersen-dp} returns the exact value of
\[
    \widetilde{\gamma}_{k}(P(6t,t)).
\]
Moreover, for every fixed $k\in\{1,2,3\}$ the computational complexity is linear in $t$.
\end{theorem}

\begin{proof}
Assume that $1\le k\le 3$.

First, we prove the correctness of Algorithm~\ref{alg:local-cost-table}. We claim that it computes the table $\mu_k$ correctly. Indeed, for every triple $(\alpha,\beta,\gamma)\in\Sigma_k^3$, Algorithm~\ref{alg:local-cost-table} examines every inner state $\tau\in\Sigma_k$ and checks exactly the two conditions in the definition of $\mu_k(\alpha,\beta,\gamma)$. Since it stores the minimum value of $c(\tau)$ from among all feasible inner states, and the value remains $\infty$ if none exists, the resulting table is exactly $\mu_k$.

We now prove the correctness of Algorithm~\ref{alg:petersen-dp}. Fix a pair $(\sigma_0,\sigma_1)\in\Sigma_k^2$. For $j\ge 1$, let $D_j(\alpha,\beta)$ denote the minimum partial weight consisting of the colored outer vertices in the blocks
\[
    B_0,B_1,\ldots,B_j
\]
and the colored inner vertices in the blocks
\[
    B_1,\ldots,B_{j-1},
\]
over all assignments whose initial outer states are $\sigma_0,\sigma_1$, whose last two outer states are $\alpha,\beta$, respectively, and for which all rainbow $k$-domination constraints in the blocks $B_1,\ldots,B_{j-1}$ are satisfied.

For $j=1$, this interpretation is immediate from the initialization
\[
    D_1(\sigma_0,\sigma_1)=c(\sigma_0)+c(\sigma_1).
\]
Assume now that this interpretation holds for some $j$ with $1\le j\le t-2$. To extend a partial solution ending with outer states $(\alpha,\beta)$ by a new outer state $\gamma$, one has to add the number of colored vertices of the new outer block, namely $c(\gamma)$, and the minimum possible number of colored inner vertices in block $B_j$, namely $\mu_k(\alpha,\beta,\gamma)$. This yields the recurrence
\[
    D_{j+1}(\beta,\gamma)=
    \min_{\alpha\in\Sigma_k} \bigl(D_j(\alpha,\beta)+c(\gamma)+\mu_k(\alpha,\beta,\gamma)\bigr).
\]
Hence, by induction with respect to $j$, the table $D_{t-1}$ computed by Algorithm~\ref{alg:petersen-dp} has the claimed meaning.

At this point all blocks have been introduced, but the first and last blocks still have to be closed around the cycle of outer blocks. Since the right neighbor of the last block is obtained from $\sigma_0$ by the shift $\rho^+$, the missing local contribution for the block $B_{t-1}$ is
\[
    \mu_k(\alpha,\beta,\rho^+(\sigma_0)).
\]
Similarly, the left neighbor of the first block is obtained from the last outer state by the shift $\rho^-$, so the missing local contribution for the block $B_0$ is
\[
    \mu_k(\rho^-(\beta),\sigma_0,\sigma_1).
\]
Therefore the value considered in the last part of Algorithm~\ref{alg:petersen-dp} is exactly the total weight of the best assignment realizing the chosen initial pair $(\sigma_0,\sigma_1)$. By Lemma~\ref{lem:petersen-state-formula}, taking the minimum over all such pairs gives exactly $\widetilde{\gamma}_{k}(P(6t,t))$.

Now we establish the time complexity of the two algorithms. Let
\[
    S=|\Sigma_k|=(k+1)^6.
\]
Algorithm~\ref{alg:local-cost-table} examines exactly $S^4$ quadruples $(\alpha,\beta,\gamma,\tau)$ of constant size, computing the table $\mu_k$. Hence, its time complexity is independent of $t$ and can be expressed as $O(S^4)$. Once $\mu_k$ is available, Algorithm~\ref{alg:petersen-dp} considers $S^2$ initial pairs, and for each of them performs a dynamic program with $t-2$ stages, each involving $O(S^3)$ transitions. Thus, its running time is $O(tS^5)$. Since $S$ depends only on $k$, which is fixed, this is equivalent to $O(t)$. The memory complexity is $O(S^3)$ for the table $\mu_k$ and $O(S^2)$ for the dynamic-programming tables. The second claim assumes that only the two required consecutive tables $D_j$ and
$D_{j+1}$ are stored. Again, in both cases this is constant for fixed $k$.
\end{proof}

\begin{remark}
For the exact computations presented in Section~\ref{sec:computations}, we implemented a specialized version of Algorithm~\ref{alg:petersen-dp} for $k=2$. In this case
\[
    |\Sigma_2|=3^6=729.
\]
The implementation exploits the $12$ symmetries of the six-row representation of a block, namely the six cyclic rotations of the rows together with the transposition of the two colors. This reduces the number of starting outer-state pairs from
\[
    3^{12}=531441
\]
to $44425$ outer representatives. Moreover, all feasible transitions between pairs of states were precomputed; their total number was $308915776$.
\end{remark}

\subsection{Computational evidence and conjectures}
\label{sec:computations}
The two algorithms presented in Section~\ref{sec:dp} were implemented in the \texttt{C++} programming language. This implementation, together with additional computations utilizing the Gurobi ILP solver \cite{gurobi2026}, led to the formulation of two conjectures stated in this section.

All the computations were carried out on a machine equipped with an AMD Ryzen 9 5950X processor with boosting disabled and $64$ GB of RAM clocked at $3600$ MHz. First, the program was run with the value of $k$ set to one and $t$ ranging from two to twenty. Since $\widetilde{\gamma}_1(G)=\gamma(G)$, it computed the values of the classical domination number for the considered family of graphs. The results match the general formulas proven in \cite{wang2015domination,fu2009domination}.

For $k=2$, the specialized version of the program was used. The precomputation phase took $5.7$ seconds. We computed the exact values of
\[
    \widetilde{\gamma}_{2}(P(6t,t))
\]
for all $2\le t\le 20$. The running time ranged from $0.022$ seconds for $t=2$ to $52124.481$ seconds for $t=20$. The experimental running times closely follow a linear trend as $t$ increases, providing empirical support for the theoretical $O(t)$ complexity of the algorithm. The raw values returned by the program are shown in Table~\ref{tbl:alg-data} of Appendix~\ref{app:petersen-data}.

While the dynamic program of Section~\ref{sec:dp} is effective for $k=2$, it is not practically usable for $k=3$. Indeed, the number of states grows from
\[
    |\Sigma_2|=3^6=729
\]
to
\[
    |\Sigma_3|=4^6=4096,
\]
and the running-time bound from Theorem~\ref{thm:petersen-dp} is
\[
    O(t|\Sigma_k|^5).
\]
Thus, already at the level of the general running-time bound, passing from $k=2$ to $k=3$ increases the state-dependent factor by
\[
    \left(\frac{4096}{729}\right)^5\approx 5599.67.
\]
Together with the substantial running times observed even for the specialized $k=2$ implementation, this made the dynamic-programming approach impractical for our $k=3$ computations. For this reason, the data for $k=3$ presented in Table~\ref{tbl:gurobi-data} of Appendix~\ref{app:petersen-data} were obtained using a Python script utilizing the Gurobi ILP solver.

We now present two conjectures based on the gathered data, beginning with the case $k=2$. The broader problem of $2$-rainbow domination in generalized Petersen graphs $P(ct,t)$ was studied in \cite{shao2019two,brezovnik2023two}. Although \cite{brezovnik2023two} also considers the singleton parameter and establishes exact results for some subfamilies of $P(ct,t)$, it does not determine $\widetilde{\gamma}_{2}(P(6t,t))$ exactly for all $t$. The values obtained using our algorithms strongly suggest the following exact formula.

\begin{conjecture}
\label{conj:petersen-two}
For every integer $t\ge 2$,
\[
\widetilde{\gamma}_{2}(P(6t,t))=
\begin{cases}
18, & \text{if } t=3,\\[1mm]
5t+1+\left\lfloor \dfrac{t}{8}\right\rfloor+\left\lfloor \dfrac{t+5}{8}\right\rfloor, & \text{if } t\neq 3.
\end{cases}
\]
\end{conjecture}

We next consider the case $k=3$. Erve\v{s} and \v{Z}erovnik \cite{erves2021petersen} determined the rainbow $3$-domination number of $P(6t,t)$ exactly for several infinite subfamilies. In particular, they proved that
\[
\widetilde{\gamma}_{3}(P(6t,t))=
\begin{cases}
6t, & \text{if } t\equiv 1,5 \pmod 6,\\
6t+3, & \text{if } t\equiv 0,2,4 \pmod 6.
\end{cases}
\]
Moreover, for $t\equiv 3\pmod 6$ they established
\[
6t<\widetilde{\gamma}_{3}(P(6t,t))\le 6t+6.
\]
For the smallest graph in the latter family they further obtained
\[
21\le \widetilde{\gamma}_{3}(P(18,3))\le 23.
\]
They conjectured that the lower bound for $t\equiv3\pmod6$ can be strengthened to $6t+3$. Some of the aforementioned results were then extended to the more general family $P(ct,t)$ in \cite{rupnik2022three}.

Our computations determine all tested instances in the remaining congruence class and suggest that the upper bound $6t+6$ is attained for every $t\equiv3\pmod6$ with $t\ge9$, while $\widetilde{\gamma}_{3}(P(18,3))=22$. Together with the previously established exact results, this leads to the following conjecture.

\begin{conjecture}
\label{conj:petersen-three}
For every integer $t\ge 2$,
\[
\widetilde{\gamma}_{3}(P(6t,t))=
\begin{cases}
22, & \text{if } t=3,\\[1mm]
6t, & \text{if } t\equiv 1,5 \pmod 6,\\[1mm]
6t+3, & \text{if } t\equiv 0,2,4 \pmod 6,\\[1mm]
6t+6, & \text{if } t\equiv 3 \pmod 6 \text{ and } t\ge 9.
\end{cases}
\]
\end{conjecture}

\begin{figure}[ht]
\centering
\begin{tikzpicture}
\begin{axis}[
    ybar,
    width=0.65\columnwidth,
    height=0.45\columnwidth,
    xlabel={$t$},
    ylabel={Time (s)},
    ymode=log,
    ymin=0.01,
    ymax=3000,
    xmin=1,
    xmax=52,
    bar width=3.5pt,
    xtick={2,6,10,14,18,22,26,30,34,38,42,46,50},
    grid=both,
    minor tick num=1,
    legend style={
        at={(0.99,0.02)},
        anchor=south east,
        font=\footnotesize
    },
    legend image code/.code={
        \draw[#1] (0cm,-0.08cm) rectangle (0.18cm,0.08cm);
    }
]

\addplot[
    fill=blue,
    bar shift=0pt
] coordinates {
(5,0.017461538314819336)
(7,0.0182497501373291)
(11,0.06741762161254883)
(13,0.10137820243835449)
(17,0.16205096244812012)
(19,0.24125242233276367)
(23,0.48899245262145996)
(25,0.36942291259765625)
(29,0.6040866374969482)
(31,1.2355997562408447)
(35,1.3408498764038086)
(37,1.596841812133789)
(41,3.2714672088623047)
(43,1.6916639804840088)
(47,2.9231016635894775)
(49,4.084986448287964)
};
\addlegendentry{$t\equiv 1,5 \pmod 6$}

\addplot[
    fill=red,
    bar shift=0pt
] coordinates {
(2,0.10309100151062012)
(4,0.2658398151397705)
(6,0.5563349723815918)
(8,0.5180351734161377)
(10,0.7587618827819824)
(12,2.1831068992614746)
(14,1.6423041820526123)
(16,3.7483980655670166)
(18,3.8818368911743164)
(20,5.182603597640991)
(22,8.351414680480957)
(24,10.319264650344849)
(26,11.608290910720825)
(28,14.767839193344116)
(30,19.143115520477295)
(32,17.34304451942444)
(34,17.58126974105835)
(36,29.400073289871216)
(38,22.75926423072815)
(40,36.43378019332886)
(42,34.43136644363403)
(44,37.33524465560913)
(46,51.62084412574768)
(48,42.75998020172119)
(50,32.237194299697876)
};
\addlegendentry{$t\equiv 0,2,4 \pmod 6$}

\addplot[
    fill=green!60!black,
    bar shift=0pt
] coordinates {
(3,0.19117426872253418)
(9,8.658386707305908)
(15,27.770633935928345)
(21,207.626131772995)
(27,355.16867327690125)
(33,451.7895953655243)
(39,881.9804801940918)
(45,1217.31192445755)
(51,1505.777232170105)
};
\addlegendentry{$t\equiv 3 \pmod 6$}

\end{axis}
\end{tikzpicture}
\caption{Running times of the ILP computations for $\widetilde{\gamma}_{3}(P(6t,t))$. The vertical axis is logarithmic.}
\label{fig:ilp-times}
\end{figure}
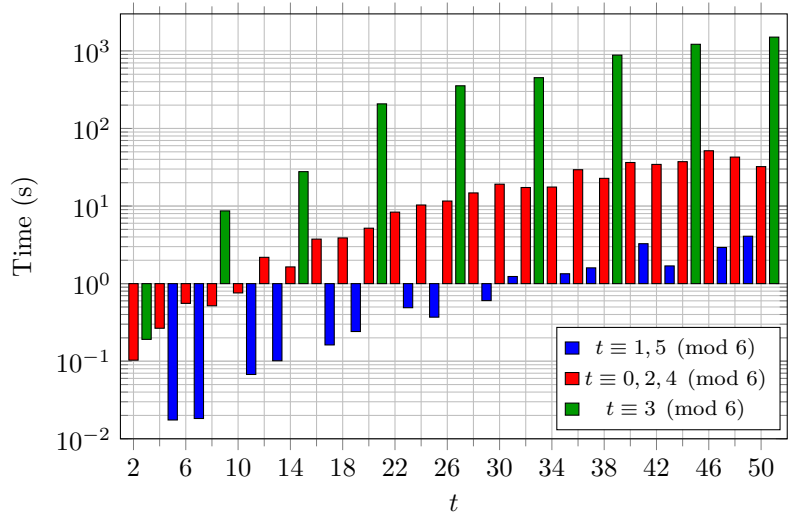

The running times of the utilized ILP solver are shown in Figure~\ref{fig:ilp-times}. They exhibit a clear dependence on the congruence class of $t$. Instances with $t\equiv 1,5\pmod 6$ are solved relatively quickly, those with even $t$ require moderately more time, and the instances with $t\equiv 3\pmod 6$ become substantially more difficult as $t$ increases. In particular, the latter instances require running times exceeding those of neighboring instances by one or more orders of magnitude.  This pattern closely reflects the distinction present in the known theoretical results: in \cite{erves2021petersen} Erve\v{s} and \v{Z}erovnik established exact values for $t\equiv 0,1,2,4,5\pmod 6$, whereas for $t\equiv 3\pmod 6$ only lower and upper bounds were obtained. Thus, the exceptional computational difficulty is present exactly in the congruence class for which the exact value had remained undetermined. Further analysis of this phenomenon might help in better understanding the structure of $P(6t, t)$ in the context of rainbow $k$-domination, and possibly in proving the conjectured exact formula.

\section{Conclusion}
\label{sec:conclusion}

We studied the rainbow $k$-domination number in regular graphs from four complementary viewpoints: general bounds for arbitrary regular graphs, complexity of the corresponding decision problem, exact values for highly symmetric regular graph families, and an exact algorithmic approach for generalized Petersen graphs.

The first part of the paper developed structural results for arbitrary regular graphs. We established two lower bounds and two upper bounds for the rainbow $k$-domination number. We then recalled the characterization of the extremal case $\widetilde{\gamma}_{r}(G)=|V(G)|/2$ using rainbow $k$-domination terminology, showing that such equality forces a bipartite structure in which one side is partitioned into $r$ color classes and every vertex on the opposite side sees exactly one vertex from each color class. Those results formed the base for the further study.

We then proved that deciding whether $\widetilde{\gamma}_{r}(G) = |V(G)|/2$ is NP-complete even when $G$ is regular and bipartite. The reduction utilized the extremal equality characterization and hypergraph coloring results of \cite{lange2026strong}. This result shows that rainbow $k$-domination remains computationally hard even under strong structural constraints.

We also presented the exact formulas for several standard regular graph families. For cycles, we noticed the coincidence of dual-server domination and rainbow $2$-domination, recalled the relevant result known for the former, and gave a self-contained proof utilizing previously established lower bounds. For M\"obius ladders, the proof involved lower bounds, counting arguments, and case-dependent modular constructions. It showed that the value of $\widetilde{\gamma}_{3}(M_n)$ is governed by the residue class of $n$ modulo $12$. For prism graphs, we determined $\widetilde{\gamma}_{k}(Pr(m))$ for all $k$, utilizing known results for classical domination and $k$-rainbow domination, together with additional counting arguments and constructions.

Finally, we developed a finite-state dynamic program for generalized Petersen graphs of the form $P(6t,t)$. For each fixed $k\in\{1,2,3\}$, the algorithm computes $\widetilde{\gamma}_{k}(P(6t,t))$ exactly in time linear in $t$. The implementation for $k=2$, together with integer linear programming computations, led to the conjectured closed formulas for $\widetilde{\gamma}_{2}(P(6t,t))$ and $\widetilde{\gamma}_{3}(P(6t,t))$ stated in Section~\ref{sec:computations}.

Several questions remain open. The most immediate problem is to prove Conjectures~\ref{conj:petersen-two} and~\ref{conj:petersen-three}. The computational evidence suggests that rainbow $k$-domination in the family $P(6t,t)$ exhibits modular behavior similar to, but more complicated than, that established for M\"obius ladders and prism graphs. It would be especially useful to turn the dynamic-programming patterns into explicit periodic constructions and corresponding lower-bound arguments.

It would also be interesting to assess the feasibility of extending the dynamic-programming framework to broader families of generalized Petersen graphs, especially graphs $P(ct,t)$ with fixed $c$ larger than $6$. Such an approach would need to carefully manage the state space to keep the calculations feasible as $c$ grows.

Another direction is to investigate the complexity of rainbow $k$-domination for other graph classes and other fixed values of $k$. In particular, it remains open whether analogous NP-completeness results can be established for regular graphs when $k<r$. It would also be interesting to determine the complexity of the exact-value problem $\widetilde{\gamma}_{k}(G)=B$ for a given integer $B$ in planar graphs, graphs subject to girth restrictions, dense graphs, and other constrained graph classes.

The results of this paper suggest that the rainbow $k$-domination number in regular graphs is influenced by an interaction between local color-forcing constraints and global structural patterns. This interaction is strong enough to produce simple exact formulas in highly structured graph families, but also complex enough to yield NP-completeness in the arbitrary regular case. Further research on this phenomenon could lead to the establishment of new structural bounds, algorithms, hardness results, and theorems regarding the relationship between rainbow $k$-domination and other domination-type parameters.

\section*{Acknowledgments}
\label{sec:acknowledgments}

This work was supported by the European Commission's Horizon Europe
Research and Innovation Programme through the Marie Sk\l odowska-Curie
Actions Staff Exchanges (MSCA-SE) under Grant Agreement No.~101182819
(COVER: Combinatorial Optimization for Versatile Applications to Emerging
Urban Problems).

The author used GitHub Copilot during the development of the algorithm and ILP implementations. After using this tool, the author reviewed and edited the content as needed and takes full responsibility for the resulting work.

\section*{Data availability}
\label{sec:data-availability}

All data generated during this study are available in Appendix~\ref{app:petersen-data}. The algorithm and ILP implementations are available online at \url{https://git.pg.edu.pl/pggit_oidc7232510/rainbow_petersen_algorithms}.

\clearpage
\appendix

\section{Proofs for M\"obius ladders}
\label{app:mobius-proofs}

\begin{proof}[Proof of Lemma~\ref{lem:mobius-upper-6}]
Write $n=12q+6$. Define a function $f:V(M_n)\to\{\emptyset,\{1\},\{2\},\{3\}\}$ such that
\[
f(v_i)=
\begin{cases}
\emptyset, & \text{if } i\equiv 0 \pmod 2,\\
\{1\}, & \text{if } i\equiv 1 \pmod 6,\\
\{2\}, & \text{if } i\equiv 3 \pmod 6,\\
\{3\}, & \text{if } i\equiv 5 \pmod 6.
\end{cases}
\]
Then exactly $n/2$ vertices are colored.

Let $i$ be even. Since
\[
\frac{n}{2}=6q+3\equiv3\pmod6,
\]
the indices of the three neighbors
\[
v_{i-1},\qquad v_{i+1},\qquad v_{i+n/2}
\]
are congruent modulo $6$ to $1$, $3$, and $5$ in some order. Therefore $f$ assigns them the three distinct colors.

Therefore $f$ is a rainbow $3$-dominating function of weight $n/2$, and consequently
\[
\widetilde{\gamma}_{3}(M_n)\le\frac{n}{2}.
\]
\end{proof}

\begin{proof}[Proof of Lemma~\ref{lem:mobius-plus-two-lower}]
Let $L=\frac{n}{2}$. Since $n\equiv2,10\pmod{12}$, Corollary~\ref{cor:cubic-extremal} gives $\widetilde{\gamma}_{3}(M_n)\ge L+1$.

Suppose, for a contradiction, that $\widetilde{\gamma}_{3}(M_n)=L+1$. Let $f$ be a minimum rainbow $3$-dominating function of $M_n$, and define
\[
C=\{v\in V(M_n):f(v)\neq\emptyset\},
\]
\[
U=\{v\in V(M_n):f(v)=\emptyset\}.
\]
Then $|U|=n-|C|=L-1$. Moreover, $U$ is an independent set. Indeed, every vertex of $U$ has three neighbors and they all belong to $C$. Since $L$ is odd, $M_n$ is bipartite with partite sets
\[
D=\{v_i:i\equiv0\pmod2\},
\]
\[
O=\{v_i:i\equiv1\pmod2\}.
\]

Partition the cycle into the $L$ disjoint edges
\[
(v_{2j},v_{2j+1}), \qquad j=0,1,\dots,L-1.
\]
Then exactly one of those pairs contains no vertex of $U$. Indeed, $U$ contains at most one vertex from each such edge and $|U|=L-1$. Relabeling if necessary, we may assume that this pair is $(v_0,v_1)$.

Therefore, exactly one of $v_{2j}$ and $v_{2j+1}$ belongs to $U$ for each $j=1,\dots,L-1$. Since $U$ is independent and $v_{2j+1}v_{2j+2}\in E(M_n)$, once the odd vertex is chosen from a pair, the even vertex cannot be chosen from the next pair. This implies that, for some $t\in\{0,\dots,L-1\}$, it holds that
\[
v_{2j}\in U \quad\text{for }1\le j\le t,
\]
\[
v_{2j+1}\in U \quad\text{for }t+1\le j\le L-1.
\]

We now claim that $t\in\{0,L-1\}$. Let $s=\frac{L-1}{2}$. Suppose that $1\le t\le s$. Then the edge
\[
v_{2t}v_{2(t+s)+1}\in E(M_n)
\]
has both endpoints in $U$, a contradiction.
Now suppose that $s<t\le L-2$. Then the edge
\[
v_{2(t+1-s)}v_{2(t+1)+1}\in E(M_n)
\]
has both endpoints in $U$, again giving a contradiction. Therefore either $t=0$ or $t=L-1$.

By symmetry, we may assume
\[
U=D\setminus\{v_0\}.
\]
Thus $v_0\in C$ is the only even vertex in $C$, and
\[
O\subseteq C.
\]

For each $i\in\{1,2,3\}$ let
\[
C_i=\{v\in V(M_n): f(v)=\{i\}\}.
\]
We count the edges between $U$ and $C_i$ in two ways. First, each vertex $u\in U$ has exactly one neighbor of color $i$, so
\[
e(U,C_i)=|U|=L-1.
\]
Second, every odd vertex has degree $3$, but the three odd neighbors of $v_0\notin U$, namely $v_1$, $v_{n-1}$, and $v_L$, have only two neighbors in $U$. Let $T=N(v_0)$ and $t_i=|C_i\cap T|$. Then
\[
e(U,C_i)=3|C_i\cap O|-t_i.
\]

Equating the two expressions gives
\[
L-1=3|C_i\cap O|-t_i,
\]
and hence
\[
t_i\equiv 1-L \pmod 3 \qquad\text{for each } i\in\{1,2,3\}.
\]

If $n\equiv 10\pmod{12}$, then $L\equiv 2\pmod 3$, so $1-L\equiv 2\pmod 3$. Since $t_i\in\{0,1,2,3\}$, this forces $t_i=2$ for each $i$. Therefore
\[
t_1+t_2+t_3=6,
\]
contradicting
\[
t_1+t_2+t_3=|T|=3.
\]
If $n\equiv 2\pmod{12}$, then $L\equiv 1\pmod 3$, so $1-L\equiv 0\pmod 3$. Thus each $t_i$ belongs to $\{0,3\}$ and $f$ assigns the same color to all three vertices in $T$. Indeed, since
\[
t_1+t_2+t_3=3,
\]
one of the numbers $t_i$ equals $3$ and the other two equal $0$. At the same time, the vertex $v_{L-1}\in U$ is adjacent to both $v_{n-1}$ and $v_L$, so it has two neighbors of the same color, contradicting $f$ being a rainbow $3$-dominating function of $M_n$. In both cases we get a contradiction. Therefore
\[
\widetilde{\gamma}_{3}(M_n)\neq \frac{n}{2}+1,
\]
and consequently
\[
\widetilde{\gamma}_{3}(M_n)\ge \frac{n}{2}+2.
\] 
\end{proof}

\begin{proof}[Proof of Lemma~\ref{lem:mobius-upper-0mod4}]
Let $L=\frac{n}{2}$. Since $n \equiv 0\pmod 4$, $L$ is even. Label
\[
V(M_n)=\{v_0,\dots,v_{n-1}\}.
\]

Define a function $f:V(M_n)\to\{\emptyset,\{1\},\{2\},\{3\}\}$ such that
\[
f(v_i)=\emptyset
\quad \text{if } 0\le i\le L-2 \text{ and } i\equiv 0\pmod 2,
\]
\[
f(v_i)=\{1\}
\quad \text{if } 0\le i\le L-1 \text{ and } i\equiv 1\pmod 6,
\]
\[
f(v_i)=\{2\}
\quad \text{if } 0\le i\le L-1 \text{ and } i\equiv 3\pmod 6,
\]
\[
f(v_i)=\{3\}
\quad \text{if } 0\le i\le L-1 \text{ and } i\equiv 5\pmod 6,
\]
\[
f(v_{L+j})=\{2\}
\quad \text{if } j\equiv 0\pmod 6 \text{ and } j\in\{0,2,\dots,L-2\},
\]
\[
f(v_{L+j})=\{3\}
\quad \text{if } j\equiv 2\pmod 6 \text{ and } j\in\{0,2,\dots,L-2\},
\]
\[
f(v_{L+j})=\{1\}
\quad \text{if } j\equiv 4\pmod 6 \text{ and } j\in\{0,2,\dots,L-2\},
\]
\[
f(v_{n-1})=\{3\},
\]
and $f(v)=\emptyset$ for all remaining vertices.

Then $f$ assigns a nonempty set to all odd-indexed vertices in $\{v_0,\dots,v_{L-1}\}$, all even-indexed vertices in $\{v_L,\dots,v_{n-2}\}$, and the single vertex $v_{n-1}$. Hence $|C|=L+1=n/2+1$.

Let $v_i$ be an uncolored vertex with $0\le i\le L-2$ and $i$ even. Its neighbors are
\[
v_{i-1},\quad v_{i+1},\quad v_{i+L},
\]
and by the definition of $f$ these three vertices receive the colors $\{1\}$, $\{2\}$, and $\{3\}$ in some order.

Now let $v_{L+j}$ be an uncolored vertex with $j$ odd and $1\le j\le L-3$. Its neighbors are
\[
v_{L+j-1},\quad v_{L+j+1},\quad v_j,
\]
and again these three vertices are assigned the colors $\{1\}$, $\{2\}$, and $\{3\}$ in some order.

Therefore every uncolored vertex sees all three colors in its neighborhood, so $f$ is a rainbow $3$-dominating function of weight $L+1$. Hence
\[
\widetilde{\gamma}_{3}(M_n)\le \frac{n}{2}+1.
\]
\end{proof}

\begin{proof}[Proof of Lemma~\ref{lem:mobius-upper-2mod12}]
Write
\[
n=12q+2 \qquad (q\ge 1),
\]
and define
\[
L=\frac{n}{2}=6q+1,
\qquad
s=\frac{L-1}{2}=3q.
\]

Define a sequence $a_0,a_1,\dots,a_{L-1}\in\{1,2,3\}$ by
\[
(a_0,a_1,a_2,a_3)=(1,2,1,2),
\]
\[
(a_{4+3r},a_{5+3r},a_{6+3r})=(3,1,2)
\qquad\text{for } r=0,1,\dots,q-2,
\]
and for $j=1,2,\dots,s$ let $a_{j+s}$ be the unique color in
\[
\{1,2,3\}\setminus\{a_{j-1},a_j\}.
\]

Let $f:V(M_n)\to\{\emptyset,\{1\},\{2\},\{3\}\}$ be a function such that
\[
f(v_{2j+1})=\{a_j\}
\qquad\text{for } j=0,1,\dots,L-1,
\]
\[
f(v_{2(s+2)})=\{1\},
\qquad
f(v_{2(s+3)})=\{1\},
\]
and $f(v)=\emptyset$ for all remaining even vertices. Then
\[
|C|=L+2=\frac{n}{2}+2.
\]

Let $v_{2j}$ be an uncolored even-indexed vertex. Its neighbors are
\[
v_{2j-1},\quad v_{2j+1},\quad v_{2j+L}=v_{2(j+s)+1},
\]
so the colors appearing in its neighborhood are
\[
\{a_{j-1},a_j,a_{j+s}\},
\]
with indices taken modulo $L$.

A direct check from the definition of the sequence shows that the three colors
\[
a_{j-1},\qquad a_j,\qquad a_{j+s}
\]
are pairwise distinct for every $j\notin\{s+2,s+3\}$, with indices taken modulo $L$. The two exceptional indices correspond to the vertices
\[
v_{2(s+2)}
\qquad\text{and}\qquad
v_{2(s+3)},
\]
which are colored. Hence every uncolored vertex sees all three colors in its neighborhood.

Therefore $f$ is a rainbow $3$-dominating function of weight $n/2+2$, and thus
\[
\widetilde{\gamma}_{3}(M_n)\le \frac{n}{2}+2.
\]
\end{proof}

\begin{proof}[Proof of Lemma~\ref{lem:mobius-upper-10mod12}]
If $n=10$, the function $f$ defined by
\[
f(v_1)=f(v_3)=\{1\},
\qquad
f(v_5)=\{2\},
\]
\[
\qquad
f(v_7)=f(v_9)=\{3\},
\qquad
f(v_2)=f(v_8)=\{1\},
\]
and
\[
f(v_0)=f(v_4)=f(v_6)=\emptyset
\]
is a valid rainbow $3$-dominating function of weight
\[
7=\frac{n}{2}+2.
\]

Assume now that
\[
n=12q+10
\]
where $q\ge 1$.
Define
\[
L=\frac{n}{2}=6q+5,
\qquad
s=\frac{L-1}{2}=3q+2.
\]

Let $b_0,b_1,\dots,b_{L-1}\in\{1,2,3\}$ be a sequence such that
\[
(b_0,b_1,b_2,b_3,b_4,b_5)=(1,2,1,3,1,2),
\]
\[
(b_{6+3r},b_{7+3r},b_{8+3r})=(3,1,2)
\qquad\text{for } r=0,1,\dots,q-2,
\]
and for $j=1,2,\dots,s$ let $b_{j+s}$ be the unique color in
\[
\{1,2,3\}\setminus\{b_{j-1},b_j\}.
\]

Define a function $f$ by coloring all odd-indexed vertices so that
\[
f(v_{2j+1})=\{b_j\}
\qquad\text{for } j=0,1,\dots,L-1,
\]
and additionally coloring the two even vertices
\[
f(v_{2(s+2)})=\{1\},
\qquad
f(v_{2(s+4)})=\{1\},
\]
while assigning $\emptyset$ to all other even vertices. Then
\[
|C|=L+2=\frac{n}{2}+2.
\]

Let $v_{2j}$ be an uncolored even-indexed vertex. Its neighbors are
\[
v_{2j-1},\quad v_{2j+1},\quad v_{2j+L}=v_{2(j+s)+1},
\]
so the colors in its neighborhood are
\[
\{b_{j-1},b_j,b_{j+s}\},
\]
with indices taken modulo $L$. A direct check from the definition of the sequence shows that the three colors are pairwise distinct for every $j\notin\{s+2,s+4\}$. The corresponding vertices
\[
v_{2(s+2)}
\qquad\text{and}\qquad
v_{2(s+4)}
\]
are colored. Hence every uncolored vertex sees all three colors in its neighborhood.

Therefore $f$ is a rainbow $3$-dominating function of weight $n/2+2$, and thus
\[
\widetilde{\gamma}_{3}(M_n)\le \frac{n}{2}+2.
\]
\end{proof}

\section{Proofs for prism graphs}
\label{app:prism-proofs}
\begin{proof}[Proof of Lemma~\ref{lem:prism-upper-2}]
Define a function $f:V(Pr(m))\to\{\emptyset,\{1\},\{2\}\}$ such that
\[
f(u_i)=
\begin{cases}
\{1\}, & \text{if } i\equiv0\pmod2,\\
\emptyset, & \text{if } i\equiv1\pmod2,
\end{cases}
\]
and
\[
f(v_i)=
\begin{cases}
\emptyset, & \text{if } i\equiv0\pmod2,\\
\{2\}, & \text{if } i\equiv1\pmod2.
\end{cases}
\]
Then the weight of $f$ is exactly $m$.

We now prove that $f$ is a rainbow $2$-dominating function of $Pr(m)$. First, consider an uncolored vertex $u_i$, where $i$ is odd. Its neighbors are $v_i$, $u_{i-1}$, and $u_{i+1}$, with indices taken modulo $m$. Since $f(u_{i-1})=f(u_{i+1})=\{1\}$ and $f(v_i)=\{2\}$, the uncolored vertex $u_i$ sees both colors in its neighborhood. Second, consider an uncolored vertex $v_i$, where $i$ is even. Then $N(v_i)=\{u_i,v_{i-1},v_{i+1}\}$. Since $f$ assigns $\{1\}$ to $u_i$ and $\{2\}$ to at least one of $v_{i-1}$ and $v_{i+1}$, an uncolored vertex $v_i$ also sees both colors in its neighborhood.

Therefore $f$ is a rainbow $2$-dominating function of weight $m$, and consequently
\[
\widetilde{\gamma}_{2}(Pr(m))\le m.
\]
\end{proof}

\begin{proof}[Proof of Lemma~\ref{lem:prism-lower-3}]
Since $Pr(m)$ is cubic of order $2m$ and $m\equiv2\pmod6$,
Corollary~\ref{cor:cubic-extremal} gives
\[
\widetilde{\gamma}_{3}(Pr(m))\ge m+1.
\]

Suppose, for a contradiction, that $\widetilde{\gamma}_{3}(Pr(m))=m+1$. Let $f$ be a minimum rainbow $3$-dominating function of $Pr(m)$, and let $U$ be its set of uncolored vertices. Then $|U|=2m-(m+1)=m-1$. Moreover, $U$ is independent. Indeed, every vertex of $U$ has degree $3$ and all three of its neighbors must be colored.

From among each pair $(u_i,v_i)$, at most one vertex belongs to $U$. Consequently, exactly one such pair contains no vertex of $U$, and each of the remaining $m-1$ pairs contains exactly one such vertex. Now consider those $m-1$ remaining pairs in cyclic order. Then, the vertices of $U$ must alternate between the two cycles of $Pr(m)$. Since $m$ is even, $Pr(m)$ is bipartite, and hence $U$ is obtained by deleting one vertex from one of its two bipartition classes.

By symmetry, let
\[
U=A\setminus\{x\},
\]
where
\[
A=\{u_i:i\equiv0\pmod2\}\cup\{v_i:i\equiv1\pmod2\},
\]
and let $B=V(Pr(m))\setminus A$. Thus every vertex of $B$ is assigned some color by $f$, while $x$ is the only colored vertex in $A$.

For each $j\in\{1,2,3\}$, let
\[
B_j=\{y\in B:f(y)=\{j\}\},
\]
and put
\[
T=N(x),\qquad t_j=|T\cap B_j|.
\]
Since every vertex of $U$ has exactly one neighbor of color $j$, it holds that $e(U,B_j)=m-1$. At the same time, every vertex of $B_j\setminus T$ has three neighbors in $U$, while every vertex of $T\cap B_j$ has two such neighbors. Hence $m-1=3|B_j|-t_j$. Consequently,
\[
t_j\equiv1-m\equiv2\pmod3.
\]
Since $0\le t_j\le3$, this forces $t_j=2$ for every $j\in\{1,2,3\}$.
Thus
\[
t_1+t_2+t_3=6,
\]
which contradicts
\[
t_1+t_2+t_3=|T|=3.
\]

Therefore a rainbow $3$-dominating function of weight $m+1$ does not
exist, and consequently
\[
\widetilde{\gamma}_{3}(Pr(m))\ge m+2.
\]
\end{proof}

\section{Gathered data}
\label{app:petersen-data}

\begin{table}[htbp]
\centering
\begin{tabular}{|c|c|r|}
\hline
$t$ & $\widetilde{\gamma}_{2}(P(6t,t))$ & Time (s) \\
\hline
2  & 11  & 0.022 \\
3  & 18  & 0.739 \\
4  & 22  & 61.838 \\
5  & 27  & 163.080 \\
6  & 32  & 2016.193 \\
7  & 37  & 5431.389 \\
8  & 43  & 9603.877 \\
9  & 48  & 13080.890 \\
10 & 53  & 16445.925 \\
11 & 59  & 20510.875 \\
12 & 64  & 23924.655 \\
13 & 69  & 27256.908 \\
14 & 74  & 30721.392 \\
15 & 79  & 34144.957 \\
16 & 85  & 37980.734 \\
17 & 90  & 41429.770 \\
18 & 95  & 44834.172 \\
19 & 101 & 48707.385 \\
20 & 106 & 52124.481 \\
\hline
\end{tabular}
\caption{Exact values of the rainbow $2$-domination number of $P(6t,t)$ and computation times gathered using Algorithm~\ref{alg:petersen-dp}.}
\label{tbl:alg-data}
\end{table}

\begin{table}[htbp]
\centering
\begin{tabular}{|c|c|r| @{\qquad} |c|c|r|}
\hline
$t$ & $\widetilde{\gamma}_{3}(P(6t,t))$ & Time (s)
&
$t$ & $\widetilde{\gamma}_{3}(P(6t,t))$ & Time (s) \\
\hline
 2 &  15 &    0.103 & 27 & 168 &  355.169 \\
 3 &  22 &    0.191 & 28 & 171 &   14.768 \\
 4 &  27 &    0.266 & 29 & 174 &    0.604 \\
 5 &  30 &    0.017 & 30 & 183 &   19.143 \\
 6 &  39 &    0.556 & 31 & 186 &    1.236 \\
 7 &  42 &    0.018 & 32 & 195 &   17.343 \\
 8 &  51 &    0.518 & 33 & 204 &  451.790 \\
 9 &  60 &    8.658 & 34 & 207 &   17.581 \\
10 &  63 &    0.759 & 35 & 210 &    1.341 \\
11 &  66 &    0.067 & 36 & 219 &   29.400 \\
12 &  75 &    2.183 & 37 & 222 &    1.597 \\
13 &  78 &    0.101 & 38 & 231 &   22.759 \\
14 &  87 &    1.642 & 39 & 240 &  881.980 \\
15 &  96 &   27.771 & 40 & 243 &   36.434 \\
16 &  99 &    3.748 & 41 & 246 &    3.271 \\
17 & 102 &    0.162 & 42 & 255 &   34.431 \\
18 & 111 &    3.882 & 43 & 258 &    1.692 \\
19 & 114 &    0.241 & 44 & 267 &   37.335 \\
20 & 123 &    5.183 & 45 & 276 & 1217.312 \\
21 & 132 &  207.626 & 46 & 279 &   51.621 \\
22 & 135 &    8.351 & 47 & 282 &    2.923 \\
23 & 138 &    0.489 & 48 & 291 &   42.760 \\
24 & 147 &   10.319 & 49 & 294 &    4.085 \\
25 & 150 &    0.369 & 50 & 303 &   32.237 \\
26 & 159 &   11.608 & 51 & 312 & 1505.777 \\
\hline
\end{tabular}
\caption{Exact values of the rainbow $3$-domination number of $P(6t,t)$ and computation times obtained using the Gurobi ILP solver for $2\le t\le51$.}
\label{tbl:gurobi-data}
\end{table}

\end{document}